\documentclass[10pt,journal]{IEEEtran}
\ifCLASSINFOpdf
\else
\fi
\usepackage[cmex10]{amsmath}
\usepackage{graphicx,subfigure}
\graphicspath{ {./figs/} }

\usepackage{amssymb}
\usepackage{amsthm}
\usepackage{amsmath}
\usepackage{algorithm}               
\usepackage{algorithmic}             
\usepackage{url}
\usepackage{color}
\usepackage{enumitem}

\usepackage{graphicx}

\usepackage{tikz,tikz-3dplot}
\usetikzlibrary{arrows}
\usetikzlibrary{calc,fit}

\usepackage{algorithmic}

\usepackage{array}
\usepackage{url}

\usepackage{color}
\usepackage{graphicx}

\usepackage{microtype}      
\usepackage{wrapfig,subfigure}

\def\ie{i.e.}
\def\eg{e.g.}
\def\st{\textrm{s.t.}}
\def\and{\textrm{and}}
\def\ri{\textrm{relint}}
\def\ex{\textrm{ext}}
\def\conv{\textrm{conv}}
\def\aff{\textrm{aff}}
\def\dim{\textrm{dim}}

\def\rank{\textrm{rank}}

\def\spann{\textrm{span}}

\def\0{\textbf{0}}
\def\1{\textbf{1}}
\def\a{\boldsymbol{a}}
\def\c{\boldsymbol{c}}

\def\v{\boldsymbol{v}}

\def\v{\boldsymbol{v}}
\def\w{\boldsymbol{w}}
\def\x{\boldsymbol{x}}
\def\y{\boldsymbol{y}}
\def\z{\boldsymbol{z}}

\def\II{\mathcal{I}}

\def\A{\mathcal{A}}
\def\C{\mathcal{C}}

\def\E{\mathcal{E}}
\def\F{\mathcal{F}}

\def\L{\mathcal{L}}

\def\N{\mathcal{N}}
\def\O{\mathcal{O}}
\def\P{\mathcal{P}}
\def\Q{\mathcal{Q}}
\def\T{\mathcal{T}}
\def\S{\mathcal{S}}
\def\X{\mathcal{X}}

\newcommand{\RR}{I\!\!R} 

\newcommand{\myparagraph}[1]{\noindent\textbf{#1.}}

\newtheorem{theorem}{Theorem}[section]
\newtheorem{problem}{Problem}
\newtheorem{lemma}{Lemma}[section]

\newtheorem{definition}{Definition}
\newtheorem{corollary}{Corollary}[section]

\begin{document}
%
\title{On Geometric Analysis of \\ Affine Sparse Subspace Clustering}

%
%
%

\author{Chun-Guang~Li$^\dag$,~Chong~You$^\dag$,~and~Ren\'{e}~Vidal
\thanks{Chun-Guang Li is with the School of Information and Communication Engineering, Beijing University of Posts and Telecommunications, Beijing 100876, P.R. China. E-mail: lichunguang@bupt.edu.cn.} 
\thanks{Chong~You and Ren\'{e}~Vidal are with the Center for Imaging Science, Johns Hopkins University, Baltimore, MD 21218, USA. E-mail: \{cyou; rvidal\}@cis.jhu.edu. } 
\thanks{$^\dag$The two first authors contributed equally to this work.} 
}

%
%

%

\maketitle

\begin{abstract}

Sparse subspace clustering (SSC) is a state-of-the-art method for segmenting a set of data points drawn from a union of subspaces into their respective subspaces. It is now well understood that SSC produces subspace-preserving data affinity under broad geometric conditions but suffers from a connectivity issue. In this paper, we develop a novel geometric analysis for a variant of SSC, named \emph{affine} SSC (ASSC), for the problem of clustering data from a union of \emph{affine} subspaces. Our contributions include a new concept called affine independence for capturing the arrangement of a collection of affine subspaces. Under the affine independence assumption, we show that ASSC is guaranteed to produce subspace-preserving affinity. Moreover, inspired by the phenomenon that the $\ell_1$ regularization no longer induces sparsity when the solution is nonnegative, we further show that subspace-preserving recovery can be achieved under much weaker conditions for all data points other than the extreme points of samples from each subspace. In addition, we confirm a curious observation that the affinity produced by ASSC may be subspace-dense---which could guarantee the subspace-preserving affinity of ASSC to produce correct clustering under rather weak conditions. We validate the theoretical findings on carefully designed synthetic data and evaluate the performance of ASSC on several real data sets.

\end{abstract}


%
\IEEEpeerreviewmaketitle


\section{Introduction}
\label{Intro}

\IEEEPARstart{I}{n} many applications that involve processing images, videos and text, high-dimensional data can be well approximated by a union of low-dimensional subspaces. \textit{Subspace clustering} is the problem of recovering the underlying low-dimensional subspaces and assigning each data point to the subspace to which it belongs \cite{Vidal:SPM11-SC}. It has found many important applications in, \eg, motion segmentation \cite{Costeira:IJCV98, Rao:PAMI10}, hybrid system identification \cite{Bako:Automatica11}, matrix completion \cite{Li:TSP16}, and genes expression profiles clustering \cite{McWilliams:DMKD14}.

Over the past decade, there has been a surge of research interests in subspace clustering and numerous algorithms have been proposed, \eg, sparse subspace clustering (SSC) \cite{Elhamifar:CVPR09, Elhamifar:TPAMI13,Dyer:JMLR13,You:CVPR16-SSCOMP}, low rank subspace clustering \cite{Liu:ICML10,Liu:TPAMI13,Favaro:CVPR11,Vidal:PRL14}, and so on \cite{Lu:ECCV12, Wang:NIPS13-LRR+SSC, Patel:ICCV13, Park:NIPS14, Li:CVPR15, Guo:IJCAI15, Wang:AISTAT16, You:CVPR16-EnSC, Li:TIP17, Xin:TSP18, Li:ICPR18}. Among them, SSC is known to enjoy both broader theoretical guarantees \cite{Elhamifar:CVPR09,Elhamifar:ICASSP10, Elhamifar:TPAMI13, Soltanolkotabi:AS12,Wang-Xu:ICML13,Soltanolkotabi:AS14,Wang:JMLR16,You:ICML15, Wang:ICML15, Tsakiris:ICML18} and superior experimental performance \cite{Elhamifar:CVPR09,Elhamifar:TPAMI13,You:CVPR16-SSCOMP, Li:CVPR15, Li:TIP17}.
Given a data matrix $X= [\x_1, \cdots,\x_N]\in\RR^{D\times N}$, SSC expresses each data point $\x_j$ as a sparse linear combination of all the other data points by solving the following optimization problem
\begin{align}
\label{eq:linear-ssc}
\begin{split}
\min_{\c_j} \| \c_j \|_1 ~~ \st  ~~~ \x_j =X \c_j, ~~~ c_{jj}=0,
\end{split}
\end{align}
where $\c_j$ is the $j$-th column of the coefficients matrix $C = [\c_1, \cdots,\c_N]\in\RR^{N\times N}$. This problem is motivated by the fact that each data point $\x_j$ in a subspace of dimension $d_\ell$ can be expressed as a linear combination of $d_\ell$ other data points in the same subspace. Therefore, it is reasonable to expect that the sparsest representation $\c_j$  selects only data points from the subspace to which $\x_j$ belongs, \ie, $c_{ij} \neq 0$ only when points $\x_i$ and $\x_j$ are in the same subspace---this is referred to as the \textit{subspace-preserving property}~\cite{Soltanolkotabi:AS12,Vidal:Springer16}.
One can then define a data affinity matrix whose $i,j$-th entry is set to $|c_{ij}|+|c_{ji}|$, and the segmentation of $X$ is obtained by applying spectral clustering \cite{vonLuxburg:StatComp2007} to this affinity.

Following the initial work \cite{Elhamifar:CVPR09}, the correctness of SSC has been well-studied in the past few years.
Specifically, it has been established that SSC is guaranteed to yield subspace-preserving solutions when subspaces are independent~\cite{Elhamifar:CVPR09,Elhamifar:TPAMI13}, disjoint~\cite{Elhamifar:ICASSP10}, or even partially overlapping~\cite{Soltanolkotabi:AS12}.
Moreover, SSC has been extended to dealing with datasets that are corrupted with outliers~\cite{Soltanolkotabi:AS12,You:CVPR17}, contaminated with noise~\cite{Wang-Xu:ICML13,Soltanolkotabi:AS14,Wang:JMLR16} or missing entries~\cite{Tsakiris:ICML18}, and preprocessed with dimension reduction techniques~\cite{Wang:ICML15}.

\myparagraph{Affine Subspace Clustering}
In many important applications of subspace clustering, the underlying subspaces do not pass through the origin, \ie, the subspaces are \emph{affine}.
In the motion segmentation problem in computer vision, for example, the feature point trajectories associated with a single rigid motion lie in an affine subspace \cite{Tomasi:IJCV92} of dimension $2$ or $3$, hence the trajectories of multiple rigid motions lie in a union of multiple \textit{affine} subspaces.
This motivates the problem of clustering affine subspaces, which can be
formally stated as follows.
\begin{problem}[\bf Affine subspace clustering]
\label{pro:problem-asc}
Let $X \in \RR^{D \times N}$ be a matrix whose columns are drawn from a union of $n$ affine subspaces of $\RR^D$, $\bigcup_{\ell=1}^n \{\mathcal {A}_\ell\}$, of dimensions $\{d_\ell < D\}_{\ell=1,\dots,n}$. The goal of affine subspace clustering is to segment the columns of $X$ into their corresponding affine subspaces.
\end{problem}
Three exemplar cases of affine subspace arrangement are illustrated in Fig.~\ref{fig:illustration-affine-subspaces}.
Note that unlike the case of linear subspaces in which all subspaces intersect at the origin, a union of affine subspaces may or may not intersect.

\myparagraph{Affine Sparse Subspace Clustering} To address the problem of affine subspace clustering, SSC is extended by
adding an affine constraint $\1^\top \c_j =1$ to~\eqref{eq:linear-ssc}, where $\1$ is the vector of all ones with appropriate dimension \cite{Elhamifar:CVPR09}. This leads to the following optimization problem:
\begin{align}
\label{eq:affine-ssc}
\min_{\c_j} \| \c_j \|_1 ~~ \st  ~~~ \x_j =X \c_j, ~~~ c_{jj}=0, ~~~ \1^\top \c_j =1.
\end{align}
In other words, \eqref{eq:affine-ssc} attempts to express each data point $\x_j$ as a sparse \emph{affine} combination of other data points. As in the case of SSC, spectral clustering is then applied to the affinity $|c_{ij}| + |c_{ji}|$. We refer to this approach as \textit{affine SSC} (ASSC). 

\begin{figure}[tb] 
\vspace{-2mm}
\centering
\subfigure[]{\includegraphics[clip=true,trim=0 0 0 0,width=0.325\columnwidth]{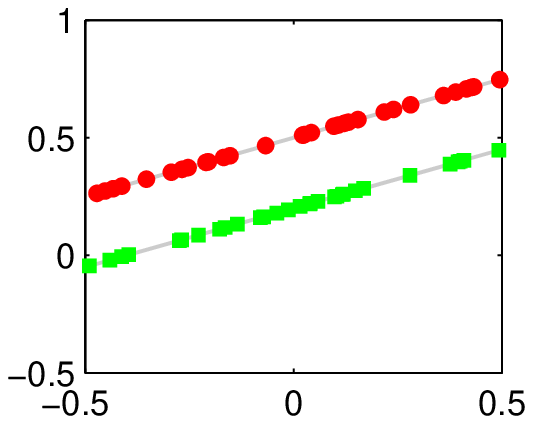}}
\subfigure[]{\includegraphics[clip=true,trim=0 0 0 0,width=0.325\columnwidth]{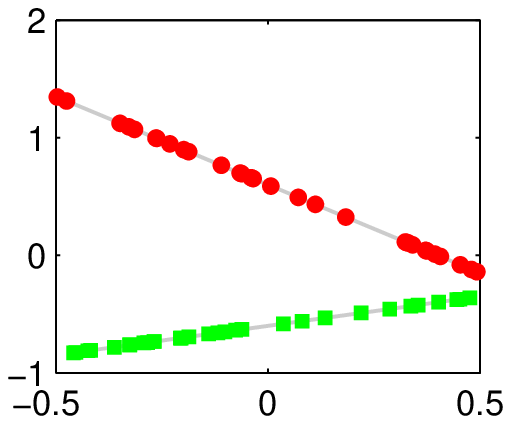}}
\subfigure[]{\includegraphics[clip=true,trim=0 0 0 0,width=0.325\columnwidth]{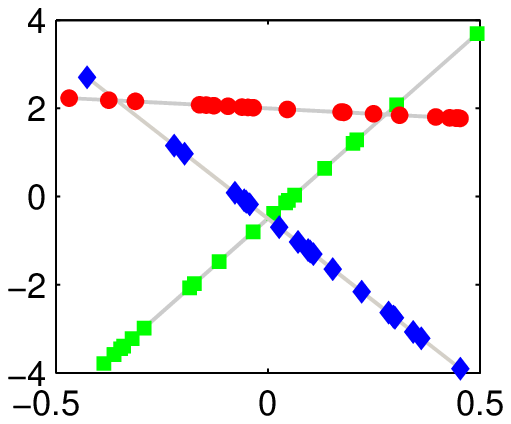}}
\caption{Illustration for data points lying in a union of affine subspaces.}
\label{fig:illustration-affine-subspaces}
\vspace{-3mm}
\end{figure}

Previous studies have shown that ASSC achieves remarkable performance in real applications \cite{Elhamifar:CVPR09, Elhamifar:TPAMI13}.
However, theoretical conditions for its correctness have rarely been considered in the literature.\footnote{In~\cite{Tsakiris:AffinePAMI17}, the theoretical results are established in the case of affine subspace for algebraic subspace clustering~\cite{Vidal:PAMI05}.}
One exception is the work \cite{Elhamifar:CVPR09}, which provides a condition under which the solution to \eqref{eq:affine-ssc} is guaranteed to be subspace-preserving.
However, this condition is characterized by the arrangement of linear subspaces that are spans of the homogeneous embedding of the original affine subspaces, making it hard to interpret.
Moreover, this result from \cite{Elhamifar:CVPR09} does not take into account the distribution of data points in each of the subspaces, making it potentially too weak.

\myparagraph{Paper Contributions}
In this work, we aim at establishing theoretical conditions for the correctness of ASSC.
Our work makes the following contributions.

\begin{itemize}[leftmargin=*]
\item \emph{Connections between SSC and ASSC.} We show that applying ASSC to a data matrix is equivalent to applying SSC to the same data matrix but in the homogeneous coordinates.
By this result, we derive correctness conditions for ASSC from existing conditions for SSC.
We argue that these conditions are expressed in the homogeneous coordinates and do not have clear geometric interpretations in the original data space.

\item \emph{Affinely independent affine subspaces.} Independence of linear subspaces is a fundamental concept in linear algebra and it is also an important assumption in the analysis of many existing 
    subspace clustering techniques \cite{Elhamifar:CVPR09,Liu:ICML10,Lu:ECCV12}.
We are not aware of an analogous concept of independence for affine subspaces in the existing literature.
Therefore, we introduce a novel concept of \emph{affine independence} for affine subspaces, and present a detailed study of its properties and geometric interpretations.
In particular, we show that ASSC produces subspace-preserving solutions if the collection of affine subspaces is affinely independent.

\item \emph{Tight conditions for interior points.} We further show that for interior points of the convex hull of data points from each subspace, there is a tight (\ie, equivalent) condition for the solution to \eqref{eq:affine-ssc} to be subspace-preserving.
This condition has weaker requirement than that the affine subspaces are affinely independent.
The analysis is based on a curious observation that the $\ell_1$ sparsity inducing regularization in ASSC becomes ineffective under the affine constraint if the optimal solution is nonnegative.\footnote{Such a phenomenon has recently been discussed in \cite{Kyrillidis:ICML13,Li:16-Simplex} for general sparse estimation problem.}

\item \emph{Provable correct clustering.} Most of the existing analysis for linear subspace clustering methods provides guarantees that the solution to \eqref{eq:linear-ssc} is subspace-preserving. However, this does not imply that the final clustering result 
    is correct since points from the same group may not be connected in the affinity graph, causing 
    an oversegmentation. Returning to ASSC, we show that for interior points of each subspace, there always exist solutions to \eqref{eq:affine-ssc} that is not only subspace-preserving but also \emph{dense}. This allows us to prove that under certain conditions, there exist solutions to \eqref{eq:affine-ssc} that 
produce correct clustering.

\item We illustrate the theoretical findings on carefully designed synthetic datasets. Moreover, we also evaluate the performance of ASSC on several real datasets and show that ASSC has better performance than SSC.
\end{itemize}

\myparagraph{Paper Outline} The remainder of this paper is organized as follows. Section~\ref{sec:affine-geometry} gives some preliminary on affine geometry. Section~\ref{sec:sufficient-condition-for-ASD} presents an analysis of the ASSC problem based on homogeneous coordinates. 
Section~\ref{sec:novel-analysis-AASC} proposes a novel geometry analysis for the ASSC problem. Section~\ref{sec:experiments} shows numerical experiments and Section~\ref{sec:conclusion} concludes the paper. 

\section{Preliminary on Affine Geometry}
\label{sec:affine-geometry}

We review some basic definitions in affine geometry.

A point $\x \in \RR^D$ is an affine combination of points $\{\x_j \in \RR^D\}_{j=1}^m$ if $\x = \sum_{j=1}^m c_j \x_j$ and $\sum_{j=1}^m c_j = 1$. 

A nonempty set $\A \subseteq \RR^D$ is an affine subspace if every affine combination of points in $\A$ lies in $\A$.
Equivalently, an affine subspace is a nonempty subset $\A \subseteq \RR^D$ of the form $\A = \x_0 + \S := \{\x_0 + \x, \x \in \S\}$, where $\S \subseteq \RR^D$ is a linear subspace and $\x_0 \in \RR^D$ is a point.
In particular, the linear subspace $\S$ is uniquely determined by $\A$ and is called the direction subspace of $\A$, denoted as $\T(\A)$.

A set of points $\{\x_j \in \RR^D\}_{j=1}^m$ is called affinely independent if $\sum_{j=1}^m c_j \x_j = \0$ and $\sum_{j=1}^m c_j = 0$ imply $c_j = 0$ for all $j \in \{1, \cdots, m\}$.

The affine hull of a data set $\X \in \RR^D$, denoted as $\aff(\X)$, is defined as the smallest affine subspace containing $\X$. Equivalently, the affine hull of $\X \in \RR^D$ is the set of all affine combinations of points in $\X$.

A set of data points $\{\x_j\}_{j=1}^m$ is said to affinely span an affine subspace $\A$ if $\aff(\{\x_j\}_{j=1}^m) = \A$. 
An affine basis of an affine subspace $\A$ is a set of affinely independent elements from $\A$ that affinely spans $\A$.

The dimension $\dim(\A)$ of an affine subspace $\A$ is defined 
by its direction subspace as $\dim(\T(\A))$. The number of points in every affine basis of an affine subspace $\A$ is $\dim(\A) + 1$.

We now return to Problem \ref{pro:problem-asc} and explain why ASSC in \eqref{eq:affine-ssc} can be used for solving it. 
Assume that for each $\ell = 1, \cdots, n$, the data matrix $X$ contains $N_\ell \gg d_\ell$ points in $\A_\ell$ and any $N_{\ell}-1$ points affinely span $\A_\ell$.
Now, consider any $\x_j\in\A_\ell$, there exist coefficients $c_{ij}$ such that $\x_j = \sum_{i\neq j: \x_i \in \A_\ell} c_{ij} \x_i$ and $ \sum_{i\neq j: \x_i \in \A_\ell} c_{ij}  = 1$.
We also set $c_{ij} = 0$ for all $i$ such that $\x_i \not\in \A_{\ell}$.
Then, the coefficients $c_{ij}$ satisfy the constraint in \eqref{eq:affine-ssc}, i.e. it has $\x_j = \sum_{i\neq j} c_{ij} \x_i$ and $ \sum_{i\neq j} c_{ij}  = 1$.

In fact, $\x_j \in \A_\ell$ can be expressed as an affine combination of at most $d_\ell + 1$ other points from its own subspace. If $d_\ell \ll N$ we have that the representation of $\x_j$ is sparse.
The primary idea of the ASSC formulation in \eqref{eq:affine-ssc} is to find such sparse representations which are also subspace-preserving.

\section{Analysis of ASSC Based on Analysis of SSC in Homogeneous Coordinates}
\label{sec:sufficient-condition-for-ASD}

In this section, we provide an analysis of ASSC based on the fact that ASSC is equivalent to SSC in homogeneous coordinates.
We first review the correctness conditions for SSC in the case of linear subspaces, and then derive correctness conditions for ASSC based on these results.

\subsection{Review of Correctness Conditions for SSC}  
\label{sec:sufficient-condition-in-Sol}

In this review subsection we assume that the columns of the data matrix $X$ lie in a union of \emph{linear} subspaces $\{ \S_\ell \}_{\ell=1}^n$ of dimensions $\{d_\ell\}_{\ell=1}^n$.
We will review conditions under which the optimal solution to problem \eqref{eq:linear-ssc} is subspace-preserving.

In \cite{Elhamifar:CVPR09}, it is shown that SSC produces subspace-preserving solution under the independent subspace assumption.
\begin{definition}[\bf Independent subspaces]
	A collection of linear subspaces $\{\S_\ell\}_{\ell=1}^n$ is said to be (linearly) independent if $\dim(\sum_{\ell=1}^n \S_\ell) = \sum_{\ell=1}^n \dim(\S_\ell)$.
\label{def:linear-independence-subspaces}
\end{definition}
Assume that $\x_j \in \S_\ell$. Let $X_{-j}^{(\ell)}$ be the submatrix of $X$ containing columns from subspace $\S_\ell$ but excluding $\x_j$.
\begin{theorem}[\cite{Elhamifar:CVPR09}]
	If the collection of subspaces $\{\S_\ell\}_{\ell=1}^n$ is independent and $\rank(X_{-j}^{(\ell)}) = d_\ell$, then every optimal solution $\c_j$ to \eqref{eq:linear-ssc} is subspace-preserving.
	\label{theorem:independent-linear}
\end{theorem}	

Theorem~\ref{theorem:independent-linear} shows that \eqref{eq:linear-ssc} produces subspace-preserving solutions under the independent subspace assumption, 
irrespective of the distribution of points in the subspace (except that $\rank(X_{-j}^{(\ell)}) = d_\ell$).
The next results from \cite{Soltanolkotabi:AS12} shows that even if the subspaces have nontrivial intersection, \eqref{eq:linear-ssc} still produces subspace-preserving solutions under certain 
separation conditions.

Let $U^{(\ell)} \in \RR^{D \times d_\ell}$ be an orthonormal basis for linear subspace $\S_\ell$. Let
$\a_j = U^{(\ell)\top} \x_j$ and $A_{-j}^{(\ell)} = U^{(\ell)\top} X_{-j}^{(\ell)}$.
Consider the following optimization problem:
\begin{align}
\label{eq:fictitious-primal}
\begin{split}
\min_{\c} ~\| \c \|_1 ~~\st  ~~~ A_{-j}^{(\ell)} \c  = \a_j.
\end{split}
\end{align}
The Lagrangian dual of problem \eqref{eq:fictitious-primal} is given by
\begin{align}
\label{eq:fictitious-dual}
\begin{split}
\max_{\w} ~~ \w^\top \a_j ~~~\st  ~~~ \| A_{-j}^{(\ell)\top} \w \|_\infty \le 1.
\end{split}
\end{align}
Let $\w_j^\ast$ be an optimal solution to problem~\eqref{eq:fictitious-dual} with minimal Euclidean norm and let $\v_j^\ast = U^{(\ell)} \w_j^\ast \in \S_{\ell}$ be the \textit{dual point} \cite{Soltanolkotabi:AS12} to $\x_j$.

We first present an important lemma which is modified from~\cite[Lemma 7.1]{Soltanolkotabi:AS12}.
\begin{lemma}(\cite{Soltanolkotabi:AS12})
\label{thm:sufficient-condition-linearSSC}
Let $\v_j^\ast$ be the dual point to $\x_j \in \S_\ell$ and $X^{(\kappa)}$ be a matrix whose columns are the data points from subspace $\S_\kappa$. Then, any optimal solution to \eqref{eq:linear-ssc} is subspace-preserving if
\begin{equation}
\label{eq:sufficient-condition-linearSSC}
\|X^{(\kappa)\top} \v_j^\ast\|_\infty < 1, ~~\forall~~ \kappa \ne \ell.
\end{equation}
\end{lemma}
The condition in Lemma~\ref{thm:sufficient-condition-linearSSC} is not particularly insightful in terms of its geometric interpretation. Under the additional assumption that the data points in $X$ are all normalized to have unit $\ell_2$ norm, one can further have an upper bound on the norm of $\v_j^\ast$ in terms of inradius and continue to derive a condition with clearer geometric interpretation.
\begin{definition} [\bf Relative inradius]
	\label{def:inradius}
	The relative inradius of a convex body $\P$, denoted by $r(\P)$, is defined as the radius of the largest Euclidean ball in the space $\spann(\P)$ that is inscribed in $\P$.
\end{definition}
Let $\X_{-j}^{(\ell)}$ be the set containing all columns of $X_{-j}^{(\ell)}$, and let $\P_{-j}^\ell := \conv(\pm \X_{-j}^{(\ell)})$ be the symmetric convex hull of $\X_{-j}^{(\ell)}$. It is shown in \cite{Soltanolkotabi:AS12} that $\| \v_j^\ast \|_2 \le {1 \over r(\P_{-j}^\ell)}$.

We further introduce the concept of subspace incoherence, and arrive at a major result in \cite{Soltanolkotabi:AS12}.
\begin{definition} [\bf Subspace incoherence]
	\label{def:subspace-incoherence}
	The subspace incoherence of a point $\x_j \in \S_\ell$ vis a vis data points in $\S_\kappa (\kappa \ne \ell)$ is defined as:
	\begin{align}
	\label{eq:subspace-incoherence}
	\mu_j := \max \{ \| X^{(\kappa)\top} \frac{\v_j^\ast}{\|\v_j^\ast\|_2} \|_\infty, ~\kappa=1,\cdots,n,~\kappa \ne \ell \},
	\end{align}
	where $\v_j^\ast$ is the dual point of $\x_j$.
\end{definition}

\begin{theorem}(\cite{Soltanolkotabi:AS12})
\label{theorem:SDP}
Suppose that the data points are all normalized to have unit $\ell_2$ norm. Then, every optimal solution to \eqref{eq:linear-ssc} is subspace-preserving if
\begin{align}
\mu_j < r(\P_{-j}^\ell).
\label{eq:SDP-condition}
\end{align}
\end{theorem}

As demonstrated in \cite{Soltanolkotabi:AS12}, the sufficient condition in Theorem~\ref{theorem:SDP} has a nice geometric interpretation. The incoherence $\mu_j$ on the LHS of \eqref{eq:SDP-condition} captures the separation between a dual direction in $\S_\ell$, $\v_j^\ast / \|\v_j^\ast\|_2$, and the points in other subspaces, $X^{(\kappa)} (\kappa \neq \ell)$. Intuitively, the incoherence is small if different subspaces are well separated. On the RHS, the inradius $r(\P_{-j}^\ell)$ captures the distribution of points in $\S_\ell$, and is expected to be large if the points are spread-out and not skewed towards a specific direction in $\S_\ell$.

\subsection{ASSC via SSC in Homogeneous Coordinates} 
\label{sec:ASSC-via-SSC-in-HE}
We now assume that the columns of the data matrix $X$ lie in a union of \emph{affine} subspaces $\{ \A_\ell \}_{\ell=1}^n$ of dimensions $\{d_\ell\}_{\ell=1}^n$, and derive conditions under which the solution $\c_j$ to the optimization problem \eqref{eq:affine-ssc} is subspace-preserving.

We first introduce the concept of homogeneous embedding.

\begin{definition}[\bf Homogeneous embedding]
\label{def:HE}
The homogeneous embedding $\hbar: \RR^D \!\rightarrow \!\RR^{D+1}$ of a point $\x \in \RR^D$  is defined as
\begin{align}
\label{eq:HE}
\hbar(\x) =\begin{bmatrix} \x \\ 1 \end{bmatrix}.
\end{align}
For convenience, we also denote $\hbar(\x)$ as $\tilde \x$.
\end{definition}
To understand why homogeneous embedding is important to the analysis of ASSC, observe that applying~\eqref{eq:linear-ssc} to the embedded data matrix $\tilde{X} := [\tilde \x_1, \cdots, \tilde \x_N]$ gives the following optimization problem:
\begin{align}
\label{eq:affine-ssc-homogeneous}
\!\min_{\c_j} \| \c_j \|_1 ~ \st  ~ \tilde\x_j := \begin{bmatrix} \x_j \\ 1 \end{bmatrix} = \begin{bmatrix} X \\ \1^\top \end{bmatrix} \c_j := \tilde{X}\c_j , ~ c_{jj}=0,
\end{align}
which is the same as the optimization problem in~\eqref{eq:affine-ssc}.
In other words, applying ASSC to data $X$ is equivalent to applying SSC to the embedded data $\tilde{X}$.
This connection between ASSC and SSC motivates us to provide theoretical justifications for ASSC by applying 
the results for SSC to embedded data points.

Before doing this, we first show that the embedded data points $\tilde{X}$ lie in a union of embedded subspaces of $\RR^{D+1}$.
Specifically, for an affine subspace $\A_\ell$, the set of embedded data points, denoted as $\tilde{\A}_{\ell}:=\{\hbar(\x): \x \in \A_\ell\}$, is an affine subspace of $\RR^{D+1}$.
In particular, the embedded affine subspace $\tilde{\A}_{\ell}$ does not pass through the origin, and therefore is not a linear subspace.
Nevertheless, $\tilde{\A}_{\ell}$ is contained in the linear subspace $\E(\A_\ell):=\spann(\tilde{\A}_{\ell})$. Therefore the embedding $\{ \tilde \A_\ell\}_{\ell=1}^n$ of the union of affine subspaces $\{\A_\ell\}_{\ell=1}^n$ is contained in the union of linear subspaces $\{\E(\A_\ell) \subseteq \RR^{D+1}\}_{\ell=1}^n$ of dimensions $\{d_\ell+1\}_{\ell=1}^n$.
Consequently, we can apply Theorem~\ref{theorem:independent-linear} to the analysis of ASSC by imposing independence assumption on the embedded linear subspaces $\{\E(\A_\ell) \subseteq \RR^{D+1}\}_{\ell=1}^n$. 
Formally, we have the following theorem which follows directly from Theorem~\ref{theorem:independent-linear}.

\begin{theorem}[\cite{Elhamifar:CVPR09}]
	If the collection of embedded subspaces $\{\E(\A_\ell) \subseteq \RR^{D+1}\}_{\ell=1}^n$ is independent and $\rank(\tilde{X}_{-j}^\ell) = \dim(\E(\A_\ell))$, then every optimal solution to \eqref{eq:affine-ssc} is subspace-preserving.
	\label{theorem:independent-linear-embedded}
\end{theorem}

Theorem \ref{theorem:independent-linear-embedded} has already appeared in \cite{Elhamifar:CVPR09} and it is also the only theoretical result for ASSC in the literature to the best of our knowledge.
While the conditions in Theorem~\ref{theorem:independent-linear-embedded} are very elegant in the embedded space, they are not very intuitive and are hard to interpret as they are not characterized by the arrangement of the affine subspaces in the original data space. In Section~\ref{sec:novel-analysis-AASC}, we will introduce a novel notion of affine independence for a collection of affine subspaces which is defined directly in the original data space.

For more general classes of embedded subspaces $\{\E(\A_\ell) \subseteq \RR^{D+1}\}_{\ell=1}^n$ which need not be independent, it is appealing to apply Theorem~\ref{theorem:SDP} for the analysis of ASSC. In doing so, a prerequisite is that all the embedded data points $\tilde{X}$ have unit $\ell_2$ norm.
This is obviously violated since the $\ell_2$ norm of the homogeneous embedding of any data point is greater than $1$, unless this data point is the origin.\footnote{It may be tempting to address this issue by redefining homogeneous embedding as appending a constant less than $1$. Nonetheless, we will still need that all the data points in $X$ to have the same $\ell_2$ norm, which is, in general, not satisfied for data in affine subspaces.}
Therefore, applying Theorem~\ref{theorem:SDP} to the embedded data $\tilde X$ is not a viable approach to the analysis of ASSC.

\subsection{Our Extension of Correctness Condition to ASSC} 
\label{sec:extension-of-the-sufficient-condition-in-Sol}

Instead of Theorem \ref{theorem:SDP}, we can derive novel correctness conditions for ASSC based on Lemma~\ref{thm:sufficient-condition-linearSSC} which does not require data points to be normalized. Specifically, by using the definition of subspace incoherence in \eqref{eq:subspace-incoherence} and applying Lemma \ref{thm:sufficient-condition-linearSSC} on the embedded data $\tilde X$, we have the following result.

\begin{theorem}
\label{theorem:ASDP-HE2}
Suppose $ \x_j \in  \A_\ell$, and let $\tilde \v^\ast_j$ be a dual point of $\tilde \x_j \in \E(\A_{\ell})$. Any optimal solution to \eqref{eq:affine-ssc} is subspace-preserving if
\begin{align}
\tilde \mu_j < \|\tilde \v^\ast_j \|^{-1}_2,
\label{eq:ASDP-condition-HE2}
\end{align}
where $\tilde \mu_j$ is the subspace incoherence of a point $\tilde \x_j \in \E(\A_{\ell})$ vis a vis data points in $\E(\A_{\kappa}) (\kappa \ne \ell)$ as defined in \eqref{eq:subspace-incoherence} with the embedded data $\tilde X$.
\end{theorem}
The quantity $\tilde \mu_j$ on the LHS of \eqref{eq:ASDP-condition-HE2} is the incoherence between the embedded data points $\{\tilde{X}_\kappa, \kappa \ne \ell\}$ and the normalized dual point $\frac{\tilde \v^\ast_j }{\|\tilde \v^\ast_j \|_2}$.
Note that unlike the analysis of SSC in \cite{Soltanolkotabi:AS12}, the data points in $\{\tilde{X}_\kappa, \kappa \ne \ell\}$ do not have unit $\ell_2$ norm.
Therefore, $\tilde \mu_j$ is affected not only by the angular distances between each point in $\{\tilde{X}_\kappa, \kappa \ne \ell\}$ and $\frac{\tilde \v^\ast_j }{\|\tilde \v^\ast_j \|_2}$, but also by the magnitude of the data points in $\{\tilde{X}_\kappa, \kappa \ne \ell\}$.
To gain some understanding of the quantity $\|\tilde \v^\ast_j \|^{-1}_2$ on the RHS of \eqref{eq:ASDP-condition-HE2}, we further have an upper bound for it as follows.
\begin{corollary}
\label{thm:length-of-dual-point}
Let $\tilde \v_j^\ast$ be the dual point of the embedded data point $\tilde\x_j$:
a) If the optimal solution $\c_j$ to \eqref{eq:affine-ssc} is nonnegative, then $\| \tilde \v_j^\ast \|^{-1}_2 \le \|\tilde \x_j \|_2$; b) If the optimal solution $\c_j$ to \eqref{eq:affine-ssc} is not nonnegative, then $\| \tilde \v_j^\ast \|^{-1}_2 < \|\tilde \x_j \|_2$.
\end{corollary}

Note that the condition \eqref{eq:ASDP-condition-HE2} can be checked in practice when the true clustering is known, since the dual point $\tilde \v_j^\ast$ can be computed by solving the optimization problem in \eqref{eq:fictitious-dual}.

While the sufficient condition \eqref{eq:ASDP-condition-HE2} seems promising, it still suffers from the drawback that it captures the properties of data points in the embedded subspaces $\{\E(\A_\ell)\}^n_{\ell=1}$ rather than that in the original affine subspaces $\{\A_\ell\}^n_{\ell=1}$, thus its interpretation in the original data space is unclear.

\section{A Novel Geometric Analysis for ASSC}
\label{sec:novel-analysis-AASC}

In this section, we provide a novel geometric analysis for the correctness of ASSC.
Our correctness conditions will be characterized by the arrangement of affine subspaces and the distribution of points on the affine subspaces, and will have clear geometric interpretations in the original data space rather than in the embedded data space. 

In Section~\ref{sec:analysis-based-on-arrangement} we develop the concept of affinely independent affine subspaces, and show that ASSC produces subspace-preserving affinities if the subspaces are affinely independent.
Section~\ref{sec:analysis-based-on-distribution} then develops tighter conditions by taking into account the relative position of data points in each affine subspace.
Finally, in Section \ref{sec:subspace-dense-solution-connectivity} we show that affinity may also be \emph{dense}, which allows us to further show that ASSC provably produce correct clustering.

\subsection{Analysis based on Subspaces Arrangement}
\label{sec:analysis-based-on-arrangement}

Recall from Section \ref{sec:sufficient-condition-in-Sol} that in linear subspace clustering, if the union of subspaces is \emph{independent}, then SSC is guaranteed to produce subspace-preserving representations.

Parallel to the analysis of SSC based on independent subspace assumption, we provide an analysis of ASSC by introducing a new concept called \textit{affinely independent} affine subspaces.
This concept has not appeared in the literature to the best of our knowledge.
Therefore, we also provide a detailed discussion on it 
before arriving at the core result of this section.

Note that the possible arrangement of a collection of affine subspaces is richer than that of a collection of linear subspaces.
This is because each affine subspace in a collection of affine subspaces can be arbitrarily positioned, while all linear subspaces in a collection of linear subspaces must all pass through a common point which is the origin of the coordinate system.
Consequently, the definition of independence for affine subspaces is not as straightforward as that for linear subspaces.
We start by considering the arrangement of two affine subspaces and introducing the notion of \emph{affinely disjoint} affine subspaces.

\begin{definition}[\bf Affinely disjoint affine subspaces]
	Two nonempty affine subspaces $\A$ and $\A'$ are said to be affinely disjoint if $\A \cap \A' = \emptyset$ and $\T(\A) \cap \T(\A') = \{\0\}$.
\label{def:affinely-disjoint-affine-subspaces}
\end{definition}
As an example, two lines (i.e., affine subspaces of dimension $1$) in $\RR^3$ are affinely disjoint if and only if they are skew lines, i.e., they are neither intersecting nor parallel.

Let $\aff(\A \cup \A')$ be the smallest affine subspace containing the two affine subspaces $\A$ and $\A'$. The notion of affine disjointness can be equivalently defined as follows.
\begin{lemma}
Two affine subspaces $\A$ and $\A'$ are affinely disjoint if and only if $\dim(\aff(\A \cup \A')) = \dim(\A) + \dim(\A')  + 1$.
\label{lemma:disjoint}
\end{lemma}
We now introduce the concept of \textit{affinely independent} affine subspaces.
\begin{definition}[\bf Affinely independent affine subspaces]
A collection of affine subspaces $\{\A_{\ell}\}_{\ell=1}^n$ is said to be affinely independent if $\dim(\aff(\cup_{\ell=1}^{n}\A_{\ell})) = \sum_{\ell=1}^n\dim(\A_{\ell}) + (n - 1)$.
\label{def:affinely-independent-affine-subspaces}
\end{definition}
It can be shown (see Lemma~\ref{thm:affine_arrangement_dimension_bound}) that for an arbitrary collection of affine subspaces $\{\A_{\ell}\}_{\ell=1}^n$, the dimension $\dim(\aff(\cup_{\ell=1}^{n}\A_{\ell}))$ is upper bounded by $\sum_{\ell=1}^n\dim(\A_{\ell}) + (n - 1)$.
Therefore, the collection $\{\A_{\ell}\}_{\ell=1}^n$ is affinely independent if the affine subspaces are arranged in a way such that the dimension of the affine hull of their union is maximized.

The notion of affinely disjoint and affinely independent subspaces are intricately related.
In the case of two subspaces, these two definitions are equivalent as can be seen from Lemma~\ref{lemma:disjoint}. In general, if a collection of affine subspaces $\{\A_{\ell}\}_{\ell=1}^n$ is affinely independent, then every two subspaces from this collection are disjoint. However, the converse of this statement is not true; one can easily construct three lines in $\RR^3$ that are pairwise affinely disjoint, but the collection of any three lines in $\RR^3$ is not affinely independent. More generally, if a set of affine subspaces $\{\A_{\ell}\}_{\ell=1}^n$ is affinely independent, then any two disjoint subsets of $\{\A_{\ell}\}_{\ell=1}^n$ are \textit{affinely disjoint}. More precisely, we have the following result.
\begin{lemma}
If a collection of affine subspaces $\{\A_{\ell}\}_{\ell=1}^n$ is affinely independent, then for any two nonempty subsets $\II, \II'\subseteq \{1, \cdots, n\}$ where $\II \cap \II' = \emptyset$, we have that the affine subspaces $\aff(\cup_{\kappa \in \II}\A_\kappa)$ and $\aff(\cup_{\kappa' \in \II'}\A_{\kappa'})$ are affinely disjoint.
\label{lemma:independent}
\end{lemma}

We now present the main result of this subsection. Assume that $\x_j \in \A_{\ell}$, and let $X_{-j}^{(\ell)}$ be the submatrix of $X$ containing all columns lying in $\A_{\ell}$ other than $\x_j$.
\begin{theorem}[\bf Subspace-preserving recovery under affine independence assumption]
If the collection of affine subspaces $\{\A_{\ell}\}_{\ell=1}^n$ is affinely independent and $\dim(\aff(X_{-j}^{(\ell)})) = d_\ell$,
then every optimal solution to \eqref{eq:affine-ssc} is subspace-preserving.
\label{theorem:affine-independent-sufficient}
\end{theorem}
\begin{proof}
From the assumption $\dim(\aff(X_{-j}^{(\ell)})) = d_\ell$, we know that problem \eqref{eq:affine-ssc} always has feasible solutions.  Suppose for the purpose of arriving at a contradiction that there is an optimal solution $\c^\ast$ to \eqref{eq:affine-ssc} which is not subspace-preserving.
Let $\N := \{i: \x_i \in \A_{\ell}\}$ be the indices of data points in the affine subspace $\A_\ell$, and $\O := \{i: \x_i \notin \A_{\ell}\}$ be the indices of data points in all other affine subspaces. From the constraint of the optimization problem \eqref{eq:affine-ssc} we have
\begin{equation}
\x_j = \sum_{i \in \N} c^\ast_i \x_i  + \sum_{i \in \O} c^\ast_i \x_i,
\label{eq:prf-subspace-preserving-1}
\end{equation}
which can be rewritten as 
\begin{equation}
\x_j - \sum_{i \in \N} c^\ast_i \x_i = \sum_{i \in \O} c^\ast_i \x_i.
\label{eq:prf-subspace-preserving-2}
\end{equation}
We consider two possible cases.
\begin{itemize}[leftmargin=*]
\item $\sum_{i \in \O} c^\ast_i = 1 - \sum_{i \in \N} c^\ast_i \ne 0$. In this case, we can divide both sides of \eqref{eq:prf-subspace-preserving-2} by $\sum_{i \in \O} c^\ast_i$, which gives the following:
	\begin{equation}
	\frac{\ \x_j - \sum_{i \in \N} c^\ast_i \x_i}{1 - \sum_{k \in \N} c^\ast_k}  = \frac{\ \ \sum_{i \in \O}c^\ast_i \x_i}{\sum_{k \in \O} c^\ast_k}.
	\label{eq:prf-subspace-preserving-3}
	\end{equation}
It is easy to check that the LHS of \eqref{eq:prf-subspace-preserving-3} is a 
point in $\A_{\ell}$ and the RHS 
of \eqref{eq:prf-subspace-preserving-3} is a 
point in $\aff(\cup_{\kappa\ne \ell} \A_{\kappa})$.
Therefore, Eq.~\eqref{eq:prf-subspace-preserving-3} asserts that $\A_{\ell} \cap \aff(\cup_{\kappa\ne \ell} \A_{\kappa}) \neq \emptyset$. 
	On the other hand, since $\{\A_{\ell}\}_{\ell=1}^n$ is affinely independent, from Lemma \ref{lemma:independent} we know that $\A_{\ell}$ and $\aff(\cup_{\kappa\ne \ell} \A_{\kappa})$ are affinely disjoint, which implies $\A_{\ell} \cap \aff(\cup_{\kappa\ne \ell} \A_{\kappa}) = \emptyset$. 
	Thus, we have arrived at a contradiction.

\item $\sum_{i \in \O} c^\ast_i = 1 - \sum_{i \in \N} c^\ast_i = 0$. In this case, the LHS 
of \eqref{eq:prf-subspace-preserving-2} is a 
point in $\T(\A_{\ell})$ and the RHS 
of \eqref{eq:prf-subspace-preserving-2} is a 
point in $\T(\aff(\cup_{\kappa\ne \ell} \A_{\kappa}))$. From the fact that $\{\A_{\ell}\}_{\ell=1}^n$ is affinely independent and Lemma \ref{lemma:independent}, we know that $\T(\A_{\ell}) \cap \T(\aff(\cup_{\kappa\ne \ell} \A_{\kappa})) =\{ \0 \}$. Therefore, we have
\begin{equation}
\x_j - \sum_{i \in \N} c^\ast_i \x_i = \sum_{i \in \O} c^\ast_i \x_i = \0.
\label{eq:prf-subspace-preserving-4}
\end{equation}
Note also that $\sum_{i \in \O} c^\ast_i=0$ implies $\sum_{i \in \N} c^\ast_i = 1$, thus we can construct a feasible solution $\bar{\c}$ to \eqref{eq:affine-ssc} in which $\bar{c}_i = c_i^*$ for all $i \in \N$ and $\bar{c}_i = 0$ for all $i \in \O$. 
Then, we see that $\|\bar{\c}\|_1 = \sum_{i \in \N} |c_i^*| < \sum_{i \in \N} |c_i^*| + \sum_{i \in \O} |c_i^*| = \|\c^*\|_1$, where the strict inequality follows from the assumption that $\c^*$ is not subspace-preserving. This contradicts with the fact that $\c^*$ is an optimal solution to \eqref{eq:affine-ssc}.
\end{itemize}
In either of the two cases we have arrived at contradictions. Thus, we have proved that $\c^*$ is subspace-preserving.
\end{proof}

Note that the condition in Theorem~\ref{theorem:affine-independent-sufficient} depends solely on the arrangement of the affine subspaces, regardless of the distribution of data points in each affine subspace (except for the condition $\dim(\aff(X_{-j}^{(\ell)})) = d_\ell$). Next, we will show that all data points other than the extreme points of each subspace have subspace-preserving solutions under much weaker conditions.

\subsection{Analysis Based on Relative Position of Data Points}
\label{sec:analysis-based-on-distribution}

Our analysis in this section is motivated by the observation that the sparsity inducing $\ell_1$ regularization in \eqref{eq:affine-ssc} becomes ineffective under the affine constraint $\1^\top \c_j = 1$ if there exists a nonnegative solution. 
In particular, such phenomenon occurs for all the data points other than the extreme points of the convex hull of data points in each subspace.
In the following, we derive novel geometric conditions which guarantee subspace-preserving solution of all such non-extreme points.

We first introduce an important concept used in our analysis.
\begin{definition}
	Let $\Q$ be a convex set. A convex set $\F \subseteq \Q$ is called a face of $\Q$ if every closed line segment in $\Q$ that has a relative interior point in $\F$ has both endpoints in $\F$.
\end{definition}
In other words, a face $\F$ is a convex subset of $\Q$ such that none of the points in $\F$ can be written as a convex combination of two points in $\Q \backslash \F$.
Geometrically, faces of a convex polytope (i.e. convex hull of a set of points) generalize the notion of vertices, edges and facets.
Specifically, zero dimensional faces of a convex set $\Q$, denoted as $\ex(\Q)$, are called the \emph{extreme points} (or vertices). They are points that cannot be expressed as a convex 
combination of any two other points from $\Q$.
One dimensional faces are called \emph{edges}, and faces of dimension one less than the dimension of $\Q$ are called \emph{facets}. Moreover, the set $\Q$ is a trivial face of $\Q$, and so is the empty set.
Finally, we note that the collection of all relative interiors of nonempty faces of $\Q$ defines a partition of $\Q$, \ie, all relative interiors in this collection are disjoint and their union is $\Q$ \cite[Theorem 18.2]{Rockafellar:70}.

Let $\Q^{(\ell)}$ be the convex hull of the columns of $X^{(\ell)}$.
According to the discussion above, for each data point $\x_j \in \A_\ell$ there is a unique face of $\Q^{(\ell)}$, say $\F_j^{(\ell)}$, such that $\x_j \in \ri(\F_j^{(\ell)})$. Now, consider the following three cases.

\begin{itemize}[leftmargin=*]
  \item The face $\F_j^{(\ell)}$ is $\Q^{(\ell)}$ itself, \ie, $\x_j$ is a relative interior point of $\Q^{(\ell)}$.
  In this case, $\x_j$ can be expressed as a convex combination of some other data points from $X^{(\ell)}$.
  We derive a subspace separation condition which is both necessary and sufficient for the solution of ASSC to be subspace-preserving.

  \item The face $\F_j^{(\ell)}$ has dimension\footnote{The dimension of a convex set is the dimension of its affine hull.} $0 < \dim(\F_j^{(\ell)}) < \dim(\Q^{(\ell)})$.
  Geometrically, $\x_j$ lies on the boundary of $\Q^{(\ell)}$ excluding the extreme points, i.e. $\x_j \in \partial \Q^{(\ell)} \backslash \ex(\Q^{(\ell)})$.
  In this case, $\x_j$ can be expressed as a convex combination of some other data points from $\F^{(\ell)}_j$. 
  We show that the subspace separation condition in the first case is still a sufficient condition.

  \item The face $\F_j^{(\ell)}$ is a zero-dimensional face (vertex) of $\Q^{(\ell)}$, \ie, $\x_j$ is an extreme point of $\Q^{(\ell)}$.
  In this case, $\x_j$ cannot be expressed as a convex combination of data points from $\A_\ell$.
  We will discuss why it can be more difficult to achieve subspace-preserving property.
\end{itemize}

\subsubsection{\bf \emph{Correctness of ASSC for Relative Interior Points}}
\label{sec:correctness-case1}

We start with an interesting observation on the optimal solutions of \eqref{eq:affine-ssc}. Since the solution 
is constrained to be affine, any feasible solution $\c_j$ to \eqref{eq:affine-ssc} must satisfy $\|\c_j\|_1 \ge \1^\top \c_j = 1$. This implies that the optimal objective value of \eqref{eq:affine-ssc} is bounded below by $1$.
Formally, we have the following result.
\begin{lemma}
	\label{theorem:assc-bound}
	The optimal value to \eqref{eq:affine-ssc} is no less than $1$.
\end{lemma}

Consider the case where $\x_j \in \ri(\Q^{(\ell)})$. Since $\Q^{(\ell)}$ is the convex hull of columns of $X^{(\ell)}$, $\x_j$ can always be expressed as a convex combination of the set of extreme points of $\Q^{(\ell)}$, which is a subset of the columns of $X_{-j}^{(\ell)}$ (recall that $X_{-j}^{(\ell)}$ denotes points in $\A_{\ell}$ excluding $\x_j$). That is, there exists a $\tilde{\c}_j^{(\ell)} \ge 0$ with $\1^\top \tilde{\c}_j^{(\ell)} = 1$ such that $\x_j = X_{-j}^{(\ell)} \tilde{\c}_j^{(\ell)}$.
Let $\tilde \c_j$ be a vector of length $N$ whose entries corresponding to $X_{-j}^{(\ell)}$ are those from $\tilde{\c}_j^{(\ell)}$ and all other entries are zero. We can see that $\tilde \c_j$ is trivially a subspace-preserving representation of $\x_j$. Moreover, $\tilde \c_j$ is a feasible solution to \eqref{eq:affine-ssc} and that it has objective value $1$. According to Lemma \ref{theorem:assc-bound}, the vector $\tilde \c_j$ that we have just constructed is an optimal solution to \eqref{eq:affine-ssc}.
In summary, we have proved the following lemma.
\begin{lemma}
\label{thm:interior-nonnegative}
	For any $\x_j \in \ri(\Q^{(\ell)})$, there always exists a subspace-preserving solution to \eqref{eq:affine-ssc}.
	Moreover, this solution is nonnegative, and the value of the objective function at the solution is $1$.
\end{lemma}

While Lemma \ref{thm:interior-nonnegative} asserts the existence of subspace-preserving solutions, it does not guarantee that all optimal solutions are subspace-preserving.
In the following, we state the major theorem of this subsection which states that all optimal solutions are subspace-preserving if and only if a subspace separation condition is satisfied.

\begin{theorem}[\bf Subspace-preserving recovery for interior points]
\label{thm:interior-subspace-preserving}
For any $\x_j \in \ri(\Q^{(\ell)})$, every optimal solution to \eqref{eq:affine-ssc} is subspace-preserving if and only if $\A_\ell$ does not intersect $\conv(\X^{(-\ell)})$, where $\X^{(-\ell)}$ is a set containing all data points in $X$ other than those from $\A_\ell$.
\end{theorem}
\begin{proof}

We first prove the ``if'' part. To prove the contrapositive, suppose that there is an optimal solution $\c^\ast$ to \eqref{eq:affine-ssc} which is not subspace-preserving. According to Lemma~\ref{thm:interior-nonnegative}, $\c^\ast$ is nonnegative and $\1^\top \c^\ast=1$.

Let $\N := \{i: \x_i \in \A_{\ell}\}$ be the indices of points in the affine subspace $\A_\ell$, and $\O := \{i: \x_i \notin \A_{\ell}\}$ be the indices of points in all other affine subspaces. We write $\x_j$ as
\begin{equation}
\x_j - \sum_{i \in \N} c^\ast_i \x_i = \sum_{i \in \O} c^\ast_i \x_i.
\label{eq:prf-interior-subspace-preserving-2}
\end{equation}
Since $\c^\ast$ is nonnegative and not subspace-preserving, we know that $\sum_{i \in \O} c^\ast_i = 1 - \sum_{i \in \N} c^\ast_i \ne 0$. Therefore, we can divide both sides of \eqref{eq:prf-interior-subspace-preserving-2} by $\sum_{i \in \O} c^\ast_i$, \ie,
\begin{equation}
\frac{\,\x_j - \sum_{i \in \N} c^\ast_i\x_i}{1 - \sum_{k \in \N} c^\ast_k}   = \frac{\ \ \sum_{i \in \O} c^\ast_i \x_i}{\sum_{k \in \O} c^\ast_k}.
\label{eq:prf-interior-subspace-preserving-3}
\end{equation}
It is easy to check that the LHS of \eqref{eq:prf-interior-subspace-preserving-3} is a 
point in $\A_{\ell}$ and the RHS of \eqref{eq:prf-interior-subspace-preserving-3} is a 
point in $\text{conv}(\X^{(-\ell)})$.
This shows that $\A_{\ell}$ and $\conv(\X^{(-\ell)})$ intersect, \ie, $\A_{\ell} \cap \conv(\X^{(-\ell)}) \neq \emptyset$, which finishes the proof by contrapositive.

For the ``only if'' part, assume for the purpose of proving by contrapositive that the affine subspace $\A_\ell$ intersects  $\conv(\X^{(-\ell)})$. Let $\z$ be any point that lies in the intersection of $\A_\ell$ and $\conv(\X^{(-\ell)})$.
From $\x_j \in \ri(\Q^{(\ell)})$, we know that $\x_j$ lies in the relative interior of $\conv(\X_{-j}^{(\ell)})$, which implies that $\x_j$ also lies in the relative interior of $\conv(\X_{-j} ^{(\ell)} \cup \{\z\})$. According to Lemma~\ref{lema:dense-convex-combination}, $\x_j$ can be expressed as a strict convex combination of points in $\X_{-j} ^{(\ell)} \cup \{\z\}$ with all coefficients being nonzero. That is,
\begin{equation}
	\x_j = \sum_{i \in \N} c_i \x_i + c_{0} \z,
	\label{eq:prf-representation-intersection}
\end{equation}
where $c_0 > 0$, $c_i > 0$ for all $i \in \N$ and $c_0 + \sum_{i \in \N} c_i = 1$.
Moreover, since $\z \in \text{conv}(\X^{(-\ell)})$, we have that $\z$ can be represented as a convex combination of points in $\text{conv}(X^{(-\ell)})$ 
\begin{equation}
	\z = \sum_{i \in \O} c_i \x_i,
\end{equation}
where $c_i \ge 0$ for all $i \in \O$ and $\sum_{i \in \O} c_i = 1$.
Substituting this representation of $\z$ into \eqref{eq:prf-representation-intersection} we get
\begin{equation}
	\x_j = \sum_{i \in \N} c_i \x_i + \sum_{i \in \O} c_{0}  c_i  \x_i,
\end{equation}
which is a representation of $\x_j$ using data points from $X$ with coefficients $c_i$ for data points $\x_i \in \A_{\ell}$ and $c_0 c_i$ for all other data points $\x_i \not \in \A_{\ell}$.
In particular, this representation is an optimal solution to \eqref{eq:affine-ssc} since it is feasible and achieves the minimum objective value $1$, but it is not subspace-preserving.

\end{proof}

Intuitively, the geometric condition in Theorem \ref{thm:interior-subspace-preserving} shows that ASSC is correct if each affine subspace $\A_{\ell}$ is separated from the convex hull of all data points from all other subspaces.
In the case where there are two affine subspaces, this condition holds if the two affine subspaces are affinely disjoint.
The converse of this statement is not true: it is easy to see that this condition can hold even if the two affine subspaces intersect (\eg, Fig.\ref{fig:illustration-affine-subspaces}(b)). In general, the condition in Theorem \ref{thm:interior-subspace-preserving} holds if the collection of affine subspaces $\{\A_{\ell}\}_{\ell=1}^n$ is affinely independent, but the converse is not true.
In other words, for relative interior points of each subspace, the correctness of ASSC can be guaranteed under the condition in Theorem \ref{thm:interior-subspace-preserving} which is weaker requirement than the condition in Theorem~\ref{theorem:affine-independent-sufficient}.

Note also that the correctness condition in Theorem \ref{thm:interior-subspace-preserving} is not only sufficient but also necessary.
It is thus quite interesting 
that the correctness of ASSC for interior points of each subspace can be exactly characterized by a simple geometric condition.

\subsubsection{\bf \emph{Correctness of ASSC for Boundary Points (excluding extreme points)}}
\label{sec:correctness-case2}

We now consider points $\x_j \in \ri (\F_{j}^{(\ell)})$ where $\F_{j}^{(\ell)}$ is a face of $\Q^{(\ell)}$ with dimension $0 < \dim(\F_{j}^{(\ell)}) < d_\ell$.
In this case, $\x_j$ can be expressed as a convex combination of extreme points of the face $\F_{j}^{(\ell)}$, which must be a subset of $X_{-j}^{(\ell)}$.
Therefore, 
with a similar argument as that for Lemma~\ref{thm:interior-nonnegative}, we have the following result.

\begin{lemma}
	For any $\x_j \in \ri (\F_{j}^{(\ell)})$ where $0 < \dim(\F_{j}^{(\ell)}) < d_\ell$, the same conclusion in Lemma~\ref{thm:interior-nonnegative} holds.
\end{lemma}

Furthermore, we have the following result whose proof is analogous to the ``if'' part of Theorem~\ref{thm:interior-subspace-preserving}.
\begin{theorem}[\bf Subspace-preserving recovery for boundary points]
\label{thm:faces-subspace-preserving-sufficient}
For any $\x_j \in \ri (\F_{j}^{(\ell)})$ where $0 < \dim(\F_{j}^{(\ell)}) < d_\ell$,
every optimal solution to \eqref{eq:affine-ssc} is subspace-preserving if $\A_\ell$ does not intersect $\conv (\X^{(-\ell)})$.
\end{theorem}
Note that the sufficient condition in Theorem~\ref{thm:faces-subspace-preserving-sufficient} is the same as the equivalent condition in Theorem \ref{thm:interior-subspace-preserving}.
Nevertheless, this sufficient condition is no longer necessary.

We now derive an alternative correctness condition. Note that if $\x_j \in \ri (\F_{j}^{(\ell)})$, then $\x_j$ also lies in the affine subspace $\aff(\F_{j}^{(\ell)})$.
By treating $\aff(\F_{j}^{(\ell)})$ as a ``fictitious'' affine subspace\footnote{In the sense that it is not one of the affine subspaces $\{\A_{\ell}\}_{\ell=1}^n$}, 
we can apply Theorem~\ref{thm:interior-subspace-preserving} to derive condition under which the solution $\c_j$ to \eqref{eq:affine-ssc} is subspace-preserving with respect to $\aff(\F_{j}^{(\ell)})$, \ie, nonzero entries of $\c_j$ only correspond to data points lying in $\aff(\F_{j}^{(\ell)})$.
Concretely, we have the following result.

\begin{corollary}
\label{thm:faces-preserving}
Assume that $\x_j \in \ri (\F_{j}^{(\ell)})$ where $0 < \dim(\F_{j}^{(\ell)}) < d_\ell$.
Every optimal solution $\c_j$ to \eqref{eq:affine-ssc} is such that an entry of $\c_j$ is nonzero only if the corresponding data point lies in $\aff(\F_{j}^{(\ell)})$, if and only if $\aff(\F_{j}^{(\ell)})$
does not intersect $\conv (\X^{(-\F_{j}^{(\ell)})})$, where $\X^{(-\F_{j}^{(\ell)})}$ is a set containing all data points other than those on $\F_{j}^{(\ell)}$.
\end{corollary}

Since the affine subspace $\aff(\F_{j}^{(\ell)})$ is a subset of $\A_\ell$, we further have the following corollary.
\begin{corollary}
\label{thm:faces-subspace-preserving}
For any $\x_j \in \ri (\F_{j}^{(\ell)})$ where $0 < \dim(\F_{j}^{(\ell)}) < d_\ell$,
every optimal solution to \eqref{eq:affine-ssc} is subspace-preserving if $\aff(\F_{j}^{(\ell)})$
does not intersect $\conv (\X^{(-\F_{j}^{(\ell)})})$.
\end{corollary}

For a comparison of Corollary~\ref{thm:faces-subspace-preserving} and Theorem~\ref{thm:faces-subspace-preserving-sufficient}, note that if $\A_{\ell} \cap \conv (\X^{(-\ell)}) = \emptyset$ then we have $\aff(\F_{j}^{(\ell)}) \cap \conv (\X^{(-\F_{j}^{(\ell)})}) = \emptyset$ (see Appendix for a proof). This shows that Corollary~\ref{thm:faces-subspace-preserving} is a tighter 
result than Theorem~\ref{thm:faces-subspace-preserving-sufficient}.
On the other hand, the condition in Theorem~\ref{thm:faces-subspace-preserving-sufficient} is characterized by the separation between one subspace $\A_{\ell}$ and points from other subspaces, making it more interpretable from the perspective of subspace clustering.

\subsubsection{\bf \emph{Correctness of ASSC for Extreme Points}}
\label{sec:correctness-case3}

An extreme point $\x_j \in \ex(\Q^{(\ell)})$ of a subspace can no longer be expressed as a convex combination of other data points from its own subspace.
Therefore, any feasible solution $\c_j$ to \eqref{eq:affine-ssc} that is subspace-preserving must be an affine combination but not a convex combination, i.e. $\c_j$ must have negative entries.
Such a representation has its $\ell_1$ norm strictly greater than $1$, and therefore does not achieve the lower bound of the objective value of \eqref{eq:affine-ssc} as 
in Lemma~\ref{theorem:assc-bound}.
Consequently, we can no longer derive results parallel to those in Lemma~\ref{thm:interior-nonnegative} and Theorem \ref{thm:interior-subspace-preserving}.
Nevertheless, we note that subspace-preserving recovery for extreme points can still be guaranteed under the affine independence assumption from Theorem~\ref{theorem:affine-independent-sufficient} or the separation condition in Theorem~\ref{theorem:ASDP-HE2}.

On the other hand, there are scenarios where we can assert that the optimal solution is not subspace-preserving based on the above observations.
In particular, if any feasible solution $\c_j$ to \eqref{eq:affine-ssc} is nonnegative, then by the affine constraint we have $\|\c_j\|_1 = 1$, which is lower than any subspace-preserving solution. Consequently, 
the optimal solution to \eqref{eq:affine-ssc} is not subspace-preserving. Formally, we have the following result.
\begin{theorem}
For $\x_j \in \ex(\Q^{(\ell)})$, 
if 
there exists a feasible solution $\c_j$ to \eqref{eq:affine-ssc} that is nonnegative, then the optimal solution to \eqref{eq:affine-ssc} is not subspace-preserving.	
\label{thm:extreme-point-subspace-preserving-necessary}
\end{theorem}
Theorem \ref{thm:extreme-point-subspace-preserving-necessary} states that any nonnegative solution $\c_j$ for $\x_j \in \ex(\Q^{(\ell)})$ to \eqref{eq:affine-ssc} is not subspace-preserving.
We now provide a concrete example where Theorem~\ref{thm:extreme-point-subspace-preserving-necessary} applies.
Consider any data point that is an extreme point of one of the subspaces but is not an extreme point of $\conv(\X)$, \ie, any $\x_j \in \cup_{\ell=1}^n \ex(\Q^{(\ell)}) \backslash \ex(\conv (\X))$ (if such set is nonempty), where $\X$ is the set of all columns of $X$.
Then, there exists a convex combination of points in $X_{-j}$ for $\x_j$, which thus serves a feasible solution $\c_j$ that is nonnegative.
According to Theorem~\ref{thm:extreme-point-subspace-preserving-necessary}, the optimal solution $\c_j$ to \eqref{eq:affine-ssc} is not subspace-preserving. More precisely, we have the following result.

\begin{corollary}
\label{thm:extreme-point-at-middle-necessary}
Any optimal solution to \eqref{eq:affine-ssc} for a data point $\x_j \in \cup_{\ell=1}^n \ex(\Q^{(\ell)} ) \backslash \ex(\conv ( \X ))$ is not subspace-preserving.
\end{corollary}

Although extreme points are less likely to have subspace-preserving solutions, we argue that such points can be very few relative to the entire data in practice.
Specifically, assume that the data points in subspace $\A_\ell$ are drawn from a Gaussian distribution.
It is shown \cite{DonohoTanner:AMS09} that the number of the extreme points of $\Q^{(\ell)}$ over all data points from $\A_{\ell}$ is $c (\log {N_\ell} )^{d_\ell-1 \over 2}$ as $N_\ell \to \infty$, 
where $c$ is a constant related to $d_\ell$. Therefore, the percentage of extreme points diminishes for any fixed $d_\ell$ when $N_\ell \to \infty$. 

\begin{figure*}[ht]
	\vspace{-0mm}
	\centering
	\subfigure[]{\includegraphics[clip=true,trim=0 0 0 0,width=0.385\columnwidth]{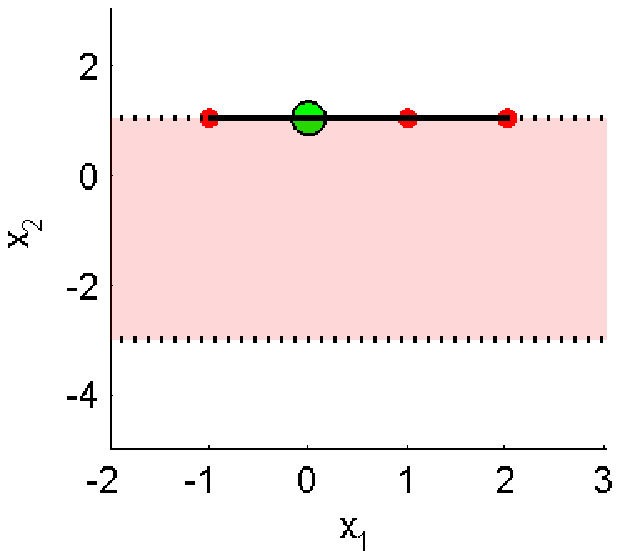}}
	\subfigure[]{\includegraphics[clip=true,trim=0 0 0 0,width=0.385\columnwidth]{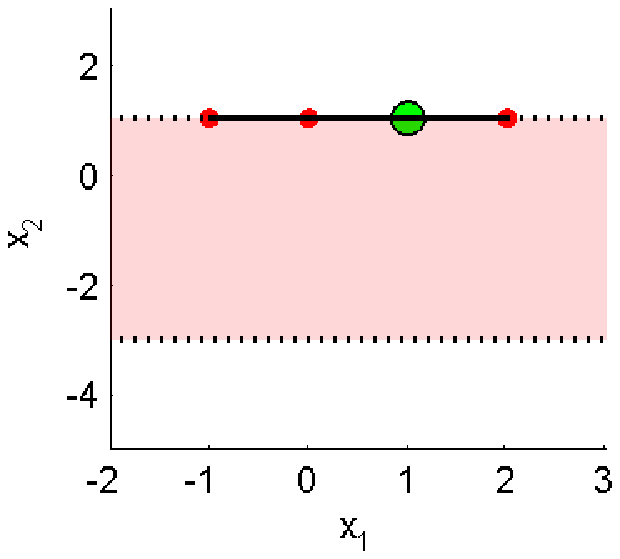}}
	\subfigure[]{\includegraphics[clip=true,trim=0 0 0 0,width=0.385\columnwidth]{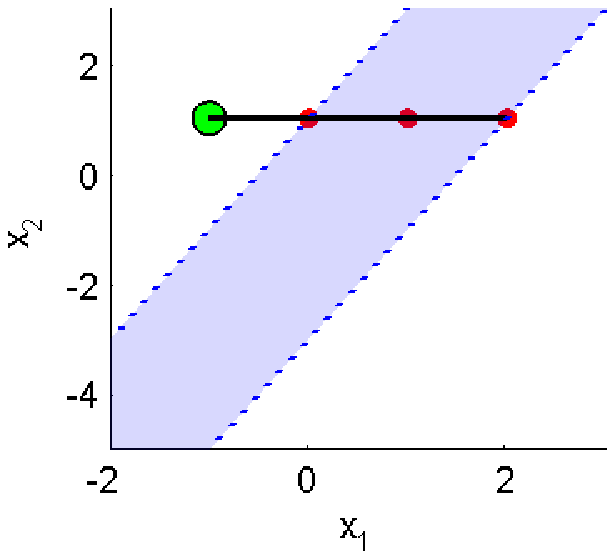}} 
	\subfigure[]{\includegraphics[clip=true,trim=0 0 0 0,width=0.385\columnwidth]{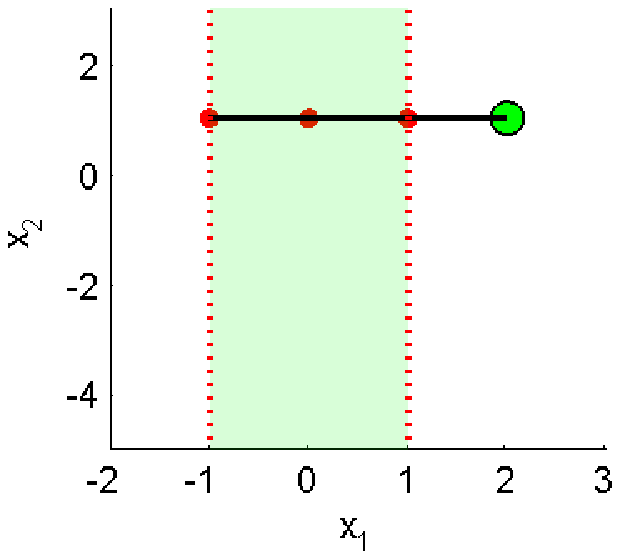}}
	\subfigure[]{\includegraphics[clip=true,trim=0 0 0 0,width=0.385\columnwidth]{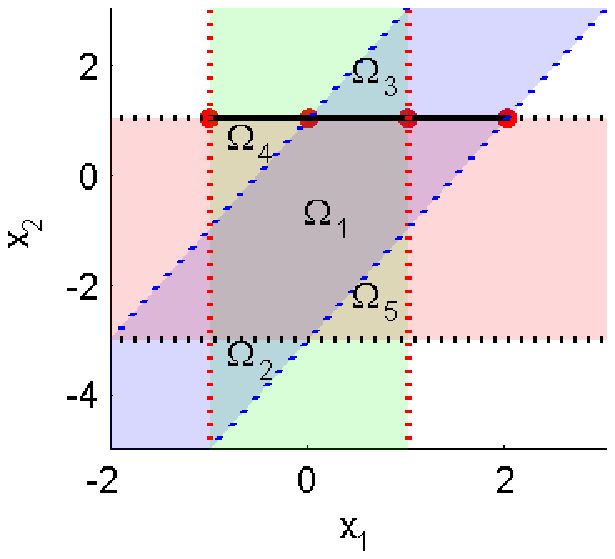}}
	\caption{Illustration for the feasible regions to guarantee subspace-preserving solutions. Panels (a) and (b) are for the interior points $\x_2=(0,1)$ and $\x_3=(1,1)$. Panels (c) and (d) are for the extreme points $\x_1=(-1, 1)$ and $\x_4=(2,1)$. Panel (e) is the feasible region taking into account for the four points.}
	\label{fig:example-feasible-region-for-subspace-preserving-solution}
	\vspace{-5mm}
\end{figure*}

\subsection{Connectivity and Provably Correct Clustering}
\label{sec:subspace-dense-solution-connectivity}

The connectivity issue \cite{Nasihatkon:CVPR11,Wang:AISTAT16} refers to the problem that even if the solution $\c_j$ is subspace-preserving for all $j \in \{1, \cdots, N\}$, there is no guarantee that each cluster form a connected component and therefore the final clustering assignment (\ie, the output of spectral clustering) may over-segment points from the same cluster into multiple clusters.
SSC is particularly prone to suffer from the connectivity issue as its affinity matrix has relatively few nonzero entries due to the sparsity of the solutions.
In fact, studies in \cite{Nasihatkon:CVPR11} showed that such over-segmentation\footnote{To address the connectivity issue, \cite{Wang:AISTAT16} proposed a post-processing procedure and proved its correctness. However, such an approach is not reliable in practice as it is very sensitive to erroneous connections in the data affinity.} can indeed happen for SSC when the dimension of subspace is greater than or equal to $4$.

It may be tempting to postulate that ASSC also suffers from the connectivity issue since it looks for sparse solutions.
However, we have a curious observation that the solution to \eqref{eq:affine-ssc} may be \emph{dense}.
Specifically, consider the case where $\x_j \in \ri (\Q^{(\ell)})$.
According to Lemma~\ref{thm:interior-nonnegative}, there exist a subspace-preserving solution to \eqref{eq:affine-ssc} that has objective value $1$.
Therefore, any subspace-preserving solution $\c_j$ to \eqref{eq:affine-ssc} will be such that $\| \c_j \|_1 = 1$.
This suggests that the regularization $\| \c_j \|_1$ in the objective of \eqref{eq:affine-ssc} no longer has any effect in selecting sparse solutions from all such subspace-preserving solutions.
In fact, we can show that \eqref{eq:affine-ssc} may produce a solution that has nonzero entries corresponding to all other data points in $\A_{\ell}$. More precisely, we have the following result whose proof can be found in the Appendix.

\begin{theorem}[\bf Existence of subspace-dense solutions]
	For data point $\x_j \in \ri (\Q^{(\ell)})$, there always exists an optimal solution $\c_j$ to \eqref{eq:affine-ssc} that is \emph{subspace-dense}, \ie, $\c_j$ is subspace-preserving and the entries of $\c_j$ corresponding to all data points in $\A_{\ell}$ other than $\x_j$ 
are nonzero.
	\label{thm:subspace-dense}
\end{theorem}

The subspace-dense solutions will properly improve the connectivity of the affinity graph, which is a property that spectral clustering desires. Unfortunately, Theorem~\ref{thm:subspace-dense} does not assert that all optimal solutions are subspace-dense.
In fact, there always exist solutions that are much sparser than a subspace-dense solution (assume $N_\ell \gg d_\ell$) since any data point $\x_j \in \ri(\Q^{(\ell)})$ can always be expressed as a convex combination of at most $d_\ell + 1$ other points from $\A_{\ell}$.
Therefore, whether a solution $\c_j$ is dense or not depends on the specific optimization algorithm for solving \eqref{eq:affine-ssc}.
In particular, if the solver to \eqref{eq:affine-ssc} is the alternating direction method of multipliers (ADMM)~\cite{Boyd:FTML10, Elhamifar:TPAMI13}, we can show that the \textit{shrinkage thresholding operation} in ADMM due to the $\ell_1$ norm will disappear 
when the optimal solution is nonnegative. In this case, it is possible that the optimal solution $\c_j$ is dense. 
We provide details of our analysis in the online Supporting Material.
In our experiments, we have observed that ADMM usually generates dense solutions for data points lying in the relative interior, \ie, $\x_j \in \ri(\Q^{(\ell)})$.

Note that for $\x_j \in \ri(\Q^{(\ell)})$ if we find a solution $\c_j$ to \eqref{eq:affine-ssc} that is subspace-dense, then it is sufficient to guarantee the connectivity of the data affinity of the data points from $\A_{\ell}$. Intuitively, the subspace-dense solutions of the relative interior points will ``pull'' all other data points, \eg, the extreme points and the boundary points, that have sparse solutions, and make the affinity graph associating all data points in the same subspace connected. We call this phenomenon a \emph{collaborative effect}. Specifically, we have the following result.
\begin{corollary}[\bf Collaborative effect] Consider $\x_j \in \A_\ell$. Suppose that $\c_j$ is a subspace-dense solution of $\x_j$ to \eqref{eq:affine-ssc}.
Then, the affinity graph induced via $|c_{ij}| + |c_{ji}|$ associating all data points $\{\x_i\}_{i \neq j}$ in subspace $\A_\ell$ is connected.
\label{corollary:collaborative-effect}
\end{corollary}

By combining the guarantees on correctness and connectivity, we have that ASSC provably produces correct clustering.
\begin{corollary}[\bf Correct clustering] Suppose that all solutions $\{\c_j\}_{j=1}^N$ to \eqref{eq:affine-ssc} are subspace-preserving and there exists at least one 
solution $\c_j$ for each subspace $\A_{\ell}$ that is subspace-dense, then ASSC produces correct clustering.
\label{corollary:correct-clustering}
\end{corollary}

Furthermore, when the collection of affine subspaces $\{\A_{\ell}\}_{\ell=1}^n$ is affinely independent, then ASSC is guaranteed to produce correct clustering as long as at least one subspace-dense solution 
for each affine subspace $\A_{\ell}$ to \eqref{eq:affine-ssc} is obtained.

\begin{corollary}[\bf Provably correct clustering under affine independence assumption]
If the collection of affine subspaces $\{\A_{\ell}\}_{\ell=1}^n$ is affinely independent and for each affine subspace $\A_{\ell}$ there exist at least one data point $\x_j$ that lies in the relative interior of $\Q^{(\ell)}$, then there always exist a set of solutions $\{\c_j\}_{j=1}^N$ to \eqref{eq:affine-ssc} such that spectral clustering on the affinity induced via  $|c_{ij}| + |c_{ji}|$ 
produces correct clustering.
\label{corollary:correct-clustering-under-affinely-independence-assumption}
\end{corollary}

When dealing with data affinity not strictly subspace-preserving, \ie, there are erroneous connections, spectral clustering is robust to those wrong connections and still give correct final clustering result. If subspace-dense solutions are found, 
the properly improved connectivity of the induced affinity graph will thus make spectral clustering more robust to wrong connections \cite{You:CVPR16-EnSC}.

\section{Experimental Evaluations}
\label{sec:experiments}

In this section, we conduct a set of experiments on carefully designed synthetic datasets to illustrate our novel geometric conditions for ASSC in Section~\ref{sec:novel-analysis-AASC}. We also test the performance of ASSC on several real world databases.

\subsection{Illustration of Theorem~\ref{theorem:ASDP-HE2}}
\label{sec:numerical-experiments}

Recall that Theorem~\ref{theorem:ASDP-HE2} is derived from applying existing results for SSC to homogeneously embedded data.	
Here, we use a numerical example to illustrate the sufficient condition in Theorem~\ref{theorem:ASDP-HE2} and to compare that with our novel geometric analysis in Section~\ref{sec:novel-analysis-AASC}.
Consider the case where there are two affine subspaces $\A_1$ and $\A_2$ in the plane $\RR^2$, where $\A_1$ is the line that passes through the following four sample points (see Fig.~\ref{fig:example-feasible-region-for-subspace-preserving-solution} (a) for an illustration):
\begin{align}
	X^{(1)} = \begin{bmatrix} \x_1~\x_2~\x_3~\x_4 \end{bmatrix} =\begin{bmatrix} -1 ~ 0 ~ 1 ~ 2 \\ ~~1 ~ 1 ~ 1 ~ 1 \end{bmatrix}.
\end{align}
We do not specify $\A_2$ and the samples on $\A_2$ for now.
For each point $\x_j, j = 1, 2, 3$ or $4$, the condition in Theorem~\ref{theorem:ASDP-HE2} can be interpreted as a geometric region such that the corresponding solution $\c_j$ is guaranteed to be subspace-preserving if 
sample points from $\A_2$ lies in this region. We now explicitly compute these regions.

We first consider the data point $\x_2 = [0, 1]^\top$.
From definition, one can compute that the dual point corresponding to $\x_2$ is given by $\tilde \v_2^\ast = [0, 0.5,0.5]^\top$.
By applying Theorem \ref{theorem:ASDP-HE2}, we know that the solution $\c_2$ is subspace-preserving if any sample point $\x = [x_{1}, x_{2}]^\top$ from $\A_2$ satisfies the following condition: $-3 < x_{2} < 1$, see Fig.~\ref{fig:example-feasible-region-for-subspace-preserving-solution} (a). For the data point $\x_3=[1, 1]^\top$, one can see that Theorem \ref{theorem:ASDP-HE2} yields the same condition as that for $\x_2$, see Fig.~\ref{fig:example-feasible-region-for-subspace-preserving-solution} (b).
Now, consider data point $\x_1 = [-1, 1]^\top$. The dual point can be computed as $\tilde \v_1^\ast = [-1, 0.5, 0.5]^\top$.
The condition in Theorem \ref{theorem:ASDP-HE2} is satisfied if for any data point $\x = [x_{1}, x_{2}] \in \A_2$ we have $-3+2x_1 < x_2 < 1 + 2 x_1$, as illustrated in Fig.~\ref{fig:example-feasible-region-for-subspace-preserving-solution} (c).
Similarly, for data point $\x_4 = [2, 1]^\top$, the dual point $\tilde \v_4^\ast$ is $[1, 0, 0]^\top$, and the corresponding geometric region is $-1 < x_{1} < 1$, as illustrated in Fig.~\ref{fig:example-feasible-region-for-subspace-preserving-solution} (d).
Overall, the region which guarantees subspace-preserving solution for all four points in $\A_1$ is the intersection of the regions above, which is $\Omega_1$. See Fig.~\ref{fig:example-feasible-region-for-subspace-preserving-solution} (e) for an illustration.

We can now compare Theorem~\ref{theorem:ASDP-HE2} with our novel geometric analysis in Section~\ref{sec:novel-analysis-AASC}.
Specifically, note that $\x_2$ and $\x_3$ are interior points of $\A_1$. From Theorem \ref{thm:interior-subspace-preserving} we know that the solution to \eqref{eq:affine-ssc} is subspace-preserving for $\x_2$ and $\x_3$ if and only if all sample points $\x \in \A_2$ either have $\x_2 > 1$ or $\x_2 < 1$, \ie, all sample points $\x \in \A_2$ lie on either side of the affine subspace $\A_1$.
Note that this condition is a much weaker requirement than the condition derived from Theorem \ref{theorem:ASDP-HE2}.
Finally, if we combine the result derived from Theorem \ref{thm:interior-subspace-preserving} for $\x_2, \x_3$ and result derived from Theorem \ref{theorem:ASDP-HE2} for $\x_1, \x_4$, then we get the conclusion that the solution to \eqref{eq:affine-ssc} is subspace-preserving for all sample points in $\A_1$ if all sample points in $\A_2$ lies in the region $\Omega_1 \cup \Omega_2 \cup \Omega_3$.

\subsection{Experiments on Synthetic Data}
\label{sec:experiments-on-toy-data}

We design four examples where two or three affine subspaces are arranged in an ambient space of $\RR^2$ or $\RR^3$, as shown in subfigure (a) of Fig.~\ref{fig:two-line-in-R3}, Fig.~\ref{fig:two-line-in-R2}, Fig.~\ref{fig:example-faces-vertex-R3}, and Fig.~\ref{fig:example-faces-vertex-R2}.
For each example, we use an ADMM algorithm to solve the optimization problem in \eqref{eq:affine-ssc} and display the representation matrix $C$ as well as the associated affinity matrix $A=\frac{1}{2}(|C|+|C^\top|)$ in subfigures (b) and (c) of the corresponding figures.
In the following, we discuss how our geometric conditions apply for the analysis of each of the four examples.

\subsubsection{Two Lines in $\RR^3$} 
\label{sec:toy-Two Lines-in-R3}

\begin{figure}[ht]
	\centering
	\vspace{-2mm}
	\subfigure[]{\includegraphics[clip=true,trim=5 0 5 0,width=0.27\columnwidth]{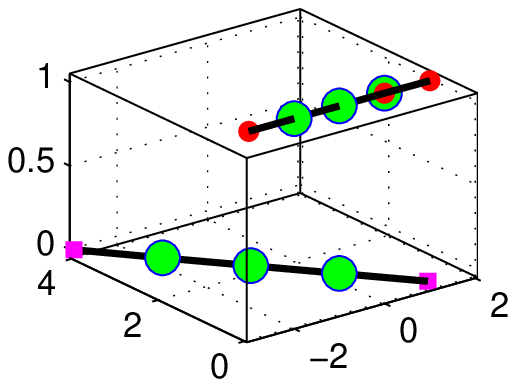}}
	\subfigure[]{\includegraphics[clip=true,trim=0 0 5 0,width=0.345\columnwidth]{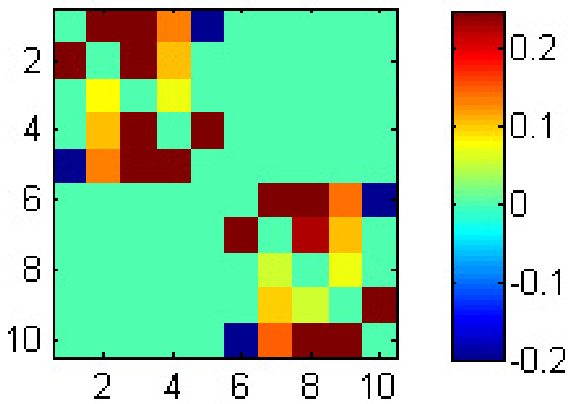}}
	\subfigure[]{\includegraphics[clip=true,trim=15 0 5 0,width=0.3\columnwidth]{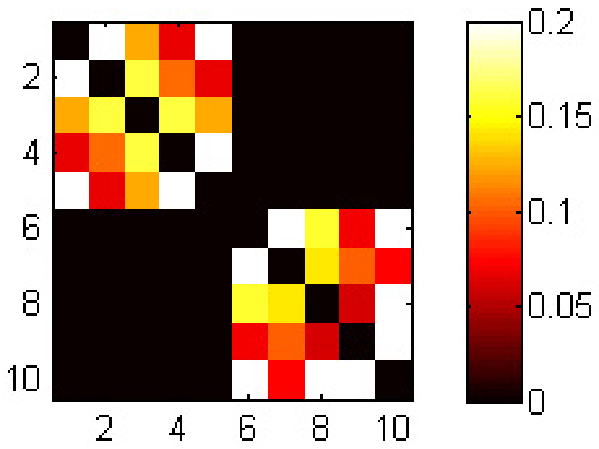}}
	\caption{Two lines in $\RR^3$. (a) Data points. (b) Coefficients matrix $C$. (c) Affinity matrix $A$.}
	\label{fig:two-line-in-R3}
	\vspace{-1mm}
\end{figure}

In this example, we sample ten data points from two line segments in $\RR^3$, denoted as $\A_1$ and $\A_2$, respectively, as illustrated in Fig.~\ref{fig:two-line-in-R3} (a), and arrange them into data matrices $X^{(1)}$ and $X^{(2)}$, where
\begin{align}
X^{(1)} &= \begin{bmatrix} \x_1~\x_2~\x_3~\x_4~\x_5 \end{bmatrix} =\begin{bmatrix} -2 ~-1 ~~~ 0 ~~~ 1 ~~~ 2 \\ ~1 ~~~~ 1 ~~~~ 1 ~~~1 ~~~ 1 \\ ~1 ~~~~ 1 ~~~~ 1 ~~~1 ~~~ 1 \end{bmatrix},~\and \nonumber \\
\!\!\!\!X^{(2)} &= \begin{bmatrix} \x_6~\x_7~\x_8~\x_9~\x_{10} \end{bmatrix} =\begin{bmatrix} ~1 ~~~ 0 ~ -1 ~ -2 ~ -3 \\ 0 ~~~ 1 ~~~~ 2 ~~~~ 3 ~~~~ 4 \\ 0 ~~~ 0 ~~~~ 0 ~~~~ 0 ~~~~ 0 \end{bmatrix}.\nonumber
\end{align}

Note that the two affine subspaces $\A_1$ and $\A_2$ are affinely disjoint and thus are also affinely independent.
According to Theorem~\ref{theorem:affine-independent-sufficient}, all points from both subspaces have subspace-preserving solutions. This is indeed what we can observe from Fig.~\ref{fig:two-line-in-R3} (b) and (c).
In addition, Theorem~\ref{thm:subspace-dense} asserts that for interior points of each subspace, there exists optimal solutions that are subspace-dense. From Fig.~\ref{fig:two-line-in-R3} (b), we can see that the solution for points $\x_2, \x_3, \x_4$ in $\A_1$ and $\x_7, \x_8, \x_9$ are indeed subspace-dense, showing that the ADMM algorithm can find such solutions in practice.

Meanwhile, comparing Fig.~\ref{fig:two-line-in-R3} (b) and (c), we can observe the \textit{collaborative effect} in the induced affinity matrix $A$ that: while the columns of coefficients matrix $C$ corresponding to the extreme points are sparse, the subspace-dense solutions of the interior points provide dense connections to the \emph{correct extreme points} and make the affinity graph properly well connected. Since the coefficients matrix $C$ is subspace-preserving and each subspace contains subspace-dense solutions, ASSC is guaranteed to generate correct clustering. In experiments, the clustering error is zero, which is exactly as predicted by Corollary~\ref{corollary:correct-clustering} and \ref{corollary:correct-clustering-under-affinely-independence-assumption}.

\subsubsection{Two Lines in $\RR^2$} 
\label{sec:toy-Two Lines-in-R2}

\begin{figure}[ht]
	\centering
	\vspace{-2mm}
	\subfigure[]{\includegraphics[clip=true,trim=5 0 5 0,width=0.27\columnwidth]{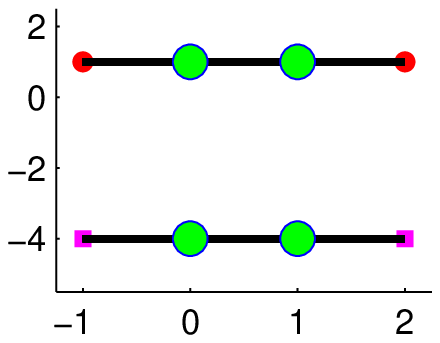}}
	\subfigure[]{\includegraphics[clip=true,trim=10 0 5 0,width=0.32\columnwidth]{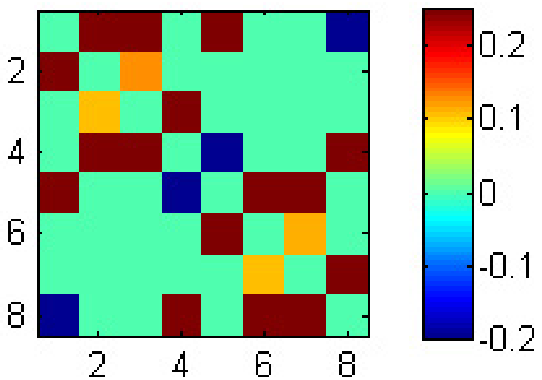}} 
	\subfigure[]{\includegraphics[clip=true,trim=15 0 5 0,width=0.30\columnwidth]{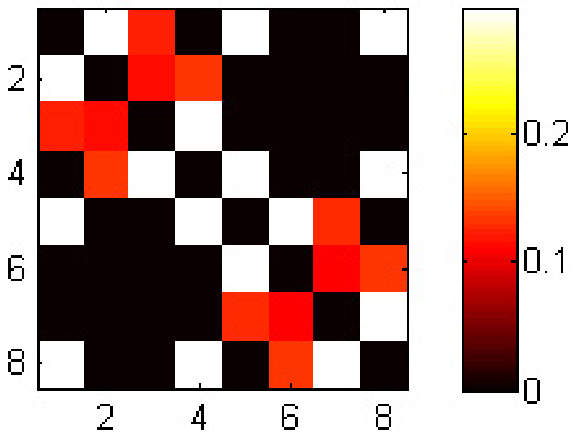}}
	\caption{Two lines in $\RR^2$. (a) Data points. (b) Coefficients matrix $C$. (c) Affinity matrix $A$.}
	\label{fig:two-line-in-R2}
	\vspace{-1mm}
\end{figure}

We sample eight data points from two affine subspaces $\A_1$ and $\A_2$, as illustrated in Fig.~\ref{fig:two-line-in-R2} (a), and arrange these data points into matrices $X^{(1)}$ and $X^{(2)}$, respectively, where
\begin{align}
X^{(1)} &= \begin{bmatrix} \x_1~\x_2~\x_3~\x_4 \end{bmatrix} =\begin{bmatrix} -1 ~ 0 ~ 1 ~ 2 \\ ~~1 ~ 1 ~ 1 ~ 1 \end{bmatrix}, ~~ \and \nonumber \\
X^{(2)} &= \begin{bmatrix} \x_5~\x_6~\x_7~\x_8 \end{bmatrix} =\begin{bmatrix} -1 ~ ~ ~ ~ 0 ~ ~ ~ ~ 1 ~ ~ ~ ~ 2 \\ -4 ~ -4 ~ -4 ~ -4 \end{bmatrix}. \nonumber
\end{align}

In this case, the two affine subspaces $\A_1$ and $\A_2$ are not affinely independent.
Still, for the relative interior points $\x_2,\x_3,\x_6$, and $x_7$, we can apply Theorem~\ref{thm:interior-subspace-preserving}, from which the solutions to \eqref{eq:affine-ssc} are asserted to be subspace-preserving. For the four extreme points (\ie, the vertices) $\x_1,\x_4,\x_5$ and $\x_8$, there is no subspace-preserving guarantee and we observe that the obtained optimal solutions are not subspace-preserving.

From Fig.~\ref{fig:two-line-in-R2} (b) we can also see that the subspace-preserving solutions for data points $\x_2,\x_3,\x_6$, and $\x_7$ are also subspace-dense.
In particular, the points $\x_2, \x_3$ (resp., $\x_6, \x_7$) are connected to the extreme points $\x_1, \x_4$ (resp., $\x_5, \x_8$) of the same subspace.
This is the \emph{collaborative effect}, which increases the connections of extreme points to data points in their respective subspaces.
Consequently, while the solutions for $\x_1, \x_4, \x_5, \x_8$ are not subspace-preserving, the improved connectivity makes spectral clustering more likely to generate correct clustering. 
As what we observed in experiments, although the affinity in Fig.~\ref{fig:two-line-in-R2} (c) has wrong connections, the final clustering error is still zero.

\subsubsection{A Triangle and a Line in $\RR^3$} 
\label{sec:toy-Triangle-Line-in-R3}

\begin{figure}[ht]
\vspace{-2mm}
\centering 
\subfigure[]{\includegraphics[clip=true,trim=5 0 5 0,width=0.25\columnwidth]{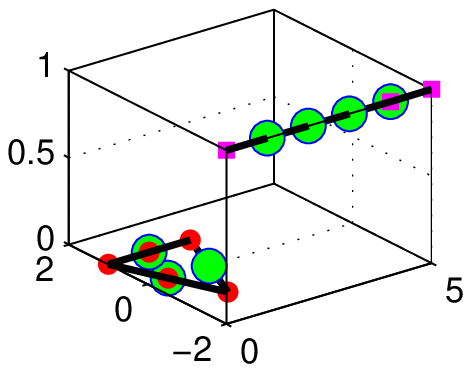}}
\subfigure[]{\includegraphics[clip=true,trim=5 0 5 0,width=0.35\columnwidth]{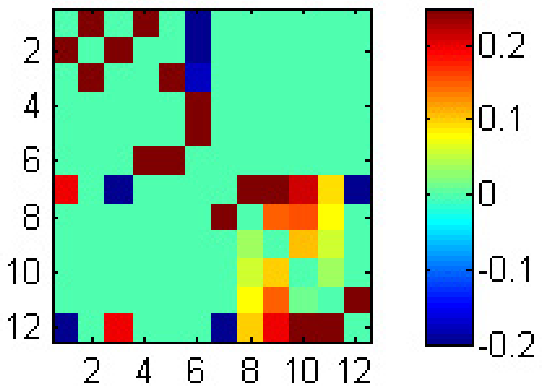}}
\subfigure[]{\includegraphics[clip=true,trim=10 0 5 0,width=0.35\columnwidth]{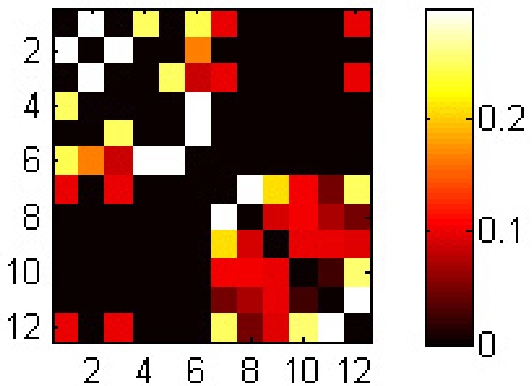}}
\caption{Data points lying in the interior of faces of convex hull $\conv(X^{(1)})$. (a) Data points. (b) Coefficients matrix $C$. (c) Affinity matrix $A$.}
\label{fig:example-faces-vertex-R3}
\vspace{-1mm}
\end{figure}


Consider a two-dimensional affine subspace $\A_1$ which is spanned by six data points (lying on a triangle) and a one-dimensional affine subspace $\A_2$ which is spanned by six data points (lying on a line segment) in $\RR^3$, as illustrated in Fig.~\ref{fig:example-faces-vertex-R3} (a). We arrange the twelve data points into matrices $X^{(1)}$ and $X^{(2)}$, respectively, where
\begin{align}
\!\!\!\! X^{(1)} &= \begin{bmatrix} \x_1~\x_2~\x_3~\x_4~\x_5~\x_6 \end{bmatrix} =\begin{bmatrix}
                  0 ~~1 ~~ 2 ~~ \frac{1}{2} ~~ \frac{3}{2} ~~ 1 \\
                 \ 1 ~~ 1 ~~ 1 ~~ 0          ~~ 0        -\!1                     \\
                 \!0 ~~ 0 ~~ 0 ~~ 0          ~~ 0    ~~ 0 \end{bmatrix}, \nonumber \\
X^{(2)} &= \begin{bmatrix} \x_7~\x_8~\x_9~\x_{10}~\x_{11}~\x_{12} \end{bmatrix} \nonumber \\
        &=\!\begin{bmatrix} ~0 ~~~~ 1 ~~~~ 2 ~~~~ 3 ~~~~ 4 ~~~~ 5 \\
                        -\!2 ~ -\!2 ~ -\!2 ~ -\!2 ~ -\!2 ~ -\!2 \\
                        ~1 ~~~~ 1 ~~~~ 1 ~~~~ 1 ~~~~ 1 ~~~~ 1 \end{bmatrix}.\nonumber
\end{align}

In this example, the two affine subspaces $\A_1$ and $\A_2$ are not affinely independent.
Moreover, for $\A_1$ there is also no relative interior points.
We can still apply 
Corollary~\ref{thm:faces-preserving} to the data points $\x_2, \x_4$, and $\x_5$ (marked in green), which are the relative interior points of the faces of $\Q^{(1)}$, and assert that the corresponding representations are subspace-preserving.
Meanwhile, we can also notice that the subspace-preserving solutions for $\x_2, \x_4$, and $\x_5$ are not subspace-dense as shown in Fig.~\ref{fig:example-faces-vertex-R3} (b).

As for $\A_2$, the data points $\x_8$, $\x_9$, $\x_{10}$, and $\x_{11}$ (marked in green) are the relative interior points of $\Q^{(2)}$. As asserted by Theorem~\ref{thm:interior-subspace-preserving}, 
the optimal solutions for these four data points are all subspace-preserving.

\subsubsection{A Triangle in $\RR^2$} 
\label{sec:toy-Triangle-in-R2}


\begin{figure}[ht]
\vspace{-2mm}
\centering
\subfigure[]{\includegraphics[clip=true,trim=0 0 5 0,width=0.25\columnwidth]{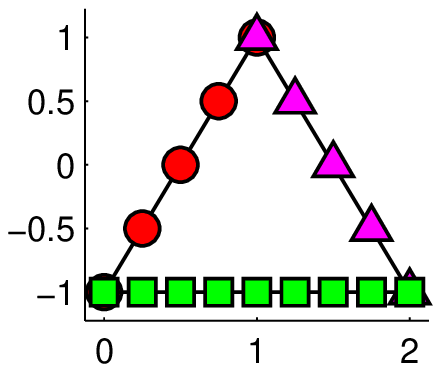}}
\subfigure[]{\includegraphics[clip=true,trim=5 0 5 0,width=0.32\columnwidth]{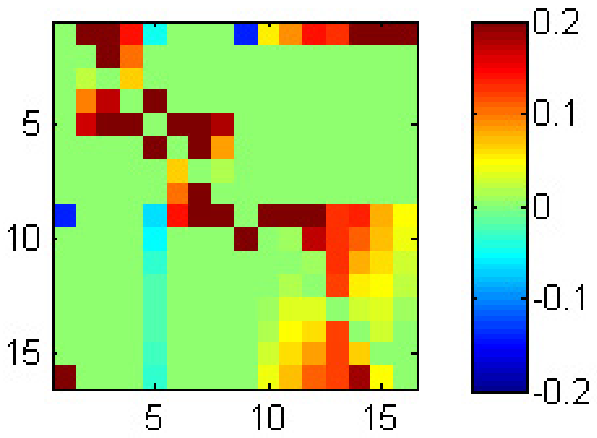}}
\subfigure[]{\includegraphics[clip=true,trim=15 0 5 0,width=0.32\columnwidth]{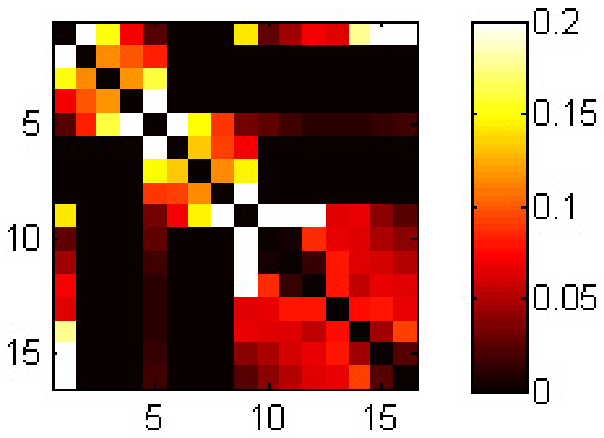}}
\caption{Data point lying in the interior of a face of the convex hull. (a) Data points. (b) Coefficients matrix $C$. (c) Affinity matrix $A$.}
\label{fig:example-faces-vertex-R2}
\vspace{-1mm}
\end{figure}

To show the special structure of the optimal solution of ASSC, we consider a two-dimensional affine subspace $\A$ which is spanned by sixteen data points lying on the three edges of a triangle in $\RR^2$, as illustrated in Fig.~\ref{fig:example-faces-vertex-R2} (a). Denote the data points lying on the three edges as $\F^{(1)}$ (red circle), $\F^{(2)}$ (purple triangle), and $\F^{(3)}$ (green square), respectively, where $\F^{(1)} \cap \F^{(3)} = \x_1 =[0,-1]^\top$, $\F^{(1)} \cap \F^{(2)} = \x_5 =[1,1]^\top$, and $\F^{(2)} \cap \F^{(3)} = \x_9 =[2,-1]^\top$.
%
Except for the three vertex points $\x_1$, $\x_5$, and $\x_9$, all other data points lie in the relative interior of the faces $\F^{(\ell)}$'s. 
As shown in Fig.~\ref{fig:example-faces-vertex-R2} (b) and (c), the relative interior points of each face yield ``face-preserving'' solutions, which are nonnegative and also ``face-dense'', showing clear collaborative effects; whereas the optimal solutions of the extreme points $\x_1$, $\x_5$, and $\x_9$ on vertices are not nonnegative.

\begin{table}[htb]
\vspace{-2mm}
\centering
\caption{Data Set Summary Information ($N$: the number of data points. $D$: ambient dimension. $n$: number of groups.)}
\label{tab:dataset}
\scriptsize
\begin{tabular}{c | c  c  c  c | c c }
\hline
		Data sets &      Iris       &   Wine  &     Wdbc    &  Ionosphere  & UCI Digits  & USPS \\
\hline
		$N$       &      150        &   178   &    569      &   351        &  390        & 2000 \\ 
		$D$       &      4          &   12    &    29       &   32         &  256        & 256  \\ 
		$n$       &      3          &   3     &    2        &   2          &  10         & 10 \\ 
\hline
\end{tabular}
\vspace{-5mm}
\end{table}

\subsection{Experiments on Real Data}

\begin{table*}[htb]
\small
\caption{\small {Clustering Error (\%) on Real Data Sets. $^\ast$: ASSC($\ell_1$) is used for USPS.}} 
\begin{center}
\footnotesize
\begin{tabular}{c | c c c c c c c  c  c}
\hline
                 &\multicolumn{2}{c}{Hopkins 155}  &Iris           &Wdbc       &Wine          &Ionosphere       &UCI Digits    &USPS           \\
                 &$n=2$&$n=3$        &$n=3$          &$n=2$      &$n=3$         &$n=2$          &$n=10$   &$n=10$ \\
\hline
SSC(n)    &5.48 &15.60        &8.00           &28.12      &34.83       &35.33      &31.87         &\!14.71 \\

ASSC(n)   &\textbf{1.95}&\textbf{4.94} &\textbf{6.67} &\textbf{8.26}&\textbf{7.87}&\textbf{29.91}&\textbf{28.67} &\textbf{11.70}$^\ast$          \\
\hline
\end{tabular}
\end{center}
\label{table:UCI-benchmark-ASSC-vs-SSC}
\vspace{-5mm}
\end{table*}

In this subsection we evaluate the performance of ASSC on real world benchmark datasets.
Since real data are often contaminated by noise and corruptions, we relax the constraint in \eqref{eq:affine-ssc} and solve the following optimization problem instead:
%
\begin{align}
\min_{c_j} \| \c_j \|_1 + \frac{\lambda}{2}\| \x_j - X\c_j\|^2_2 ~~ \st ~ \1^\top \c_j =1, ~c_{jj} =0,
\end{align}
in which the parameter $\lambda$ is a trade-off parameter.
We follow \cite{Elhamifar:TPAMI13} and set $\lambda =  \frac{\alpha}{\min_j \max_{i:i \neq j} \{\x_i^\top \x_j\}}$ and tune $\alpha$ for each dataset.
We refer to this method as ASSC(n).


We compare ASSC with a corresponding noisy version of SSC, referred to as SSC(n), which solves the following optimization problem:
\begin{align}
\min_{c_j} \| \c_j \|_1 + \frac{\lambda}{2}\| \x_j - X\c_j\|^2_2 ~~ \st ~ ~c_{jj} =0.
\end{align}
The parameter $\lambda$ in SSC(n) is set in the same way as for ASSC(n).

To evaluate ASSC, we use Hopkins 155 database \cite{Tron:CVPR07}, which is a motion segmentation database containing 155 video sequences of $2$ or $3$ motions in each sequence. Moreover, we also use five datasets from the UCI machine learning repository \cite{Dua:17}, including Iris, Wdbc, Ionosphere, Wine and Digits. In addition, we use a subset of USPS dataset \cite{Hull:PAMI94}, which consists of randomly selected 200 samples per digit. Some statistics on the size, number of clusters and dimension of these datasets are provided in Table \ref{tab:dataset}. For each dataset, we apply SSC and ASSC, and report the averaged clustering error over 10 trials.

Experimental results are reported in Table \ref{table:UCI-benchmark-ASSC-vs-SSC}.
We observe that ASSC 
achieves better performance than SSC.
In particular, on Hopkins 155, Wdbc and Wine datasets, the clustering error of ASSC is almost three times lower 
than that of SSC. These results confirm that ASSC is an effective approach for 
clustering these datasets, and may suggest that a union of affine subspaces is a better modeling of data than a union of linear subspaces.

\section{Conclusion}
\label{sec:conclusion}

We have extended the correctness conditions of SSC which are suitable for linear subspaces to the case of affine subspaces and discussed the limitations.
We have presented the concept of affine independence and established that the optimal solution is guaranteed to be subspace-preserving if the affine subspaces are affinely independent.
%
We have derived a set of geometric conditions, showing that data points not lying on the vertex of the convex hull of each subspace are easier to yield subspace-preserving solutions.
%
Furthermore, we have shown that the data points lying in the relative interior of the convex hull of each subspace could yield subspace-dense solutions, which could thus properly improve the connectivity of the affinity graph.
%
Experiments on synthetic data and real data have validated our theoretical findings and evaluated the performance of ASSC.
%

Finally, we note that the affine subspace clustering problem can also be handled by adding an affine constraint into other subspace clustering algorithms such as \cite{Liu:TPAMI13, Lu:ECCV12}, and the tools that we developed 
for the analysis of ASSC may also be useful in developing theoretical analysis of other affine subspace clustering algorithms such as \cite{Xin:TSP18, Xu:PAA17} as well as algorithms for manifold learning and clustering such as \cite{Li:ACCV09,Elhamifar:NIPS11}.

\appendix

\subsection{Proof of Corollary \ref{thm:length-of-dual-point}} 
\label{sec:appendix-for-section-III}


%
\begin{proof}
a) When the solution is nonnegative, we have that $\| \c_j \|_1 = \1^\top \c_j =1$. From the strong duality, we have that $\tilde \w^\top \tilde \a_j = \| \c_j \|_1 =1$. By the definition of dual point, we have:
$   \| \tilde \v_j \|_2
  = \| \tilde U^{(\ell)} \tilde \w_j \|_2
  = \| \tilde \w_j \|_2
\ge {| \tilde \w^\top \tilde \a_j| \over \| \tilde \a_j\|_2}
= {1 \over \| \tilde \a_j\|_2}
  = {1 \over \|\tilde \x_j \|_2}.$
%
Thus, $\| \tilde \v_j^\ast \|^{-1}_2 \le \|\tilde \x_j \|_2$.

b) When the solution is not nonnegative, we have that $\| \c_j \|_1 > \1^\top \c_j =1$. From the strong duality, we have that $\tilde \w^\top \tilde \a_j = \| \c_j \|_1 >1$. 
Then, we have:
$   \| \tilde \v_j \|_2
  = \| \tilde U^{(\ell)} \tilde \w_j \|_2
  = \| \tilde \w_j \|_2
\ge {| \tilde \w^\top \tilde \a_j| \over \| \tilde \a_j\|_2}
> {1 \over \| \tilde \a_j\|_2}
  = {1 \over \|\tilde \x_j \|_2}.$
%
Thus, $\| \tilde \v_j^\ast \|^{-1}_2 < \|\tilde \x_j \|_2$.

\end{proof}


\subsection{Proof for Lemma \ref{lemma:disjoint}}

We start with two important lemmas. 
\begin{lemma}[\cite{Gallier:GeoMethodsAppCompSciEng01}, Lemma 2.11.3]
	Let $\A$ and $\A'$ be two arbitrary affine subspaces.
	\begin{itemize}
		\item If $\A \cap \A' = \emptyset$, then $\dim(\aff(\A \cup \A')) = \dim(\A) + \dim(\A') + 1 - \dim(\T(\A) \cap \T(\A'))$;
		\item If $\A \cap \A' \ne \emptyset$, then $\dim(\aff(\A \cup \A')) = \dim(\A) + \dim(\A') - \dim(\A \cap \A')$.
	\end{itemize}
	\label{thm:affine_arrangement_dimension}
\end{lemma}
%
\begin{lemma}
	Let $\A$ and $\A'$ be two arbitrary nonempty affine subspaces. Then it has
	\begin{equation}
	\dim(\aff(\A \cup \A')) \le \dim(\A) + \dim(\A')  + 1.
	\end{equation}
	In general, let $\{\A_{\ell}\}_{\ell=1}^n$ be a collection of nonempty affine subspaces. It has
	\begin{equation}
    \dim(\aff(\cup_{\ell=1}^{n}\A_{\ell})) \le \sum_{\ell=1}^n\dim(\A_{\ell}) + (n - 1).
	\end{equation}
	\label{thm:affine_arrangement_dimension_bound}
\end{lemma}

%
Now we are ready to start the proof of Lemma \ref{lemma:disjoint}.
\begin{proof}
	For the ``only if'' part, note from Lemma~\ref{thm:affine_arrangement_dimension} that since $\A \cap \A' = \emptyset$, it has $\dim(\aff(\A \cup \A')) = \dim(A) + \dim(A') + 1 - \dim(\T(\A) \cap \T(\A'))$. Moreover, since $\T(\A) \cap \T(\A') = \{\0\}$, it has $\dim(\T(\A) \cap \T(\A')) = 0$. Therefore, we get $\dim(\aff(\A \cup \A')) = \dim(A) + \dim(A') + 1$.
	

For the ``if'' part, assume for the purpose of arriving at a contradiction that $\A \cap \A' \ne \emptyset$. From Lemma~\ref{thm:affine_arrangement_dimension} it has $\dim(\aff(\A \cup \A')) = \dim(A) + \dim(A') - \dim(\A \cap \A')$. Note that $\dim(\aff(\A \cup \A')) = \dim(A) + \dim(A') + 1$. Thus, it must have $\dim(\A \cap \A') = -1$, which is a contradiction. Therefore, we have $\A \cap \A' = \emptyset$. Then, by using Lemma~\ref{thm:affine_arrangement_dimension} again, we see that $\dim(\aff(\A \cup \A')) = \dim(A) + \dim(A') + 1 - \dim(\T(\A) \cap \T(\A'))$. Thus, we have $\dim(\T(\A) \cap \T(\A'))=0$, which implies that $\T(\A) \cap \T(\A') = \{\0\}$. By using the definition of affinely disjoint, the proof for ``if'' part is completed.

\end{proof}

\subsection{Proof for Lemma~\ref{lemma:independent}}
\begin{proof}
	Take any two nonempty sets $\II, \II' \subseteq \{1, \cdots, n\}$ such that $\II \cap \II' = \emptyset$. We show that $\aff(\cup_{\kappa \in \II}\A_\kappa)$ and $\aff(\cup_{\kappa' \in \II'}\A_{\kappa'})$ are disjoint by using contrapositive.

Assume that 
$\aff(\cup_{\kappa \in \II}\A_\kappa)$ and $\aff(\cup_{\kappa' \in \II'}\A_{\kappa'})$ are not affinely disjoint, \ie, $\dim(\aff(\cup_{\kappa \in \II}\A_{\kappa}) + \aff(\cup_{\kappa' \in \II'}\A_{\kappa'})) \ne \dim(\aff(\cup_{\kappa \in \II}\A_{\kappa})) + \dim(\aff(\cup_{\kappa' \in \II'}\A_{\kappa'})) + 1$. From Lemma~\ref{thm:affine_arrangement_dimension_bound} we know that $\dim(\aff(\cup_{\kappa \in \II}\A_{\kappa}) + \aff(\cup_{\kappa' \in \II'}\A_{\kappa'})) \le \dim(\aff(\cup_{\kappa \in \II}\A_{\kappa})) + \dim(\aff(\cup_{\kappa' \in \II'}\A_{\kappa'})) + 1$. Therefore, it must have
$\dim(\aff(\cup_{\kappa \in \II}\A_{\kappa}) + \aff(\cup_{\kappa' \in \II'}\A_{\kappa'})) < \dim(\aff(\cup_{\kappa \in \II}\A_{\kappa})) + \dim(\aff(\cup_{\kappa' \in \II'}\A_{\kappa'})) + 1. 
$

Let $\E = \II \setminus (\II' \cup \II'')$. By using Lemma~\ref{thm:affine_arrangement_dimension_bound}, we further have
\begin{equation}
\begin{split}
		&\,\dim(\aff(\cup_{\ell=1}^{n}\A_{\ell})) \\
    \le &\, \dim(\aff(\cup_{\kappa \in \II \cup \II'}\A_\kappa)) + \dim(\aff(\cup_{\kappa \in \E}\A_\kappa)) + 1\\
	  < & \,\dim(\aff(\cup_{\kappa \in \II}\A_\kappa)) + \dim(\aff(\cup_{\kappa \in \II'}\A_\kappa)) + 1 \\
        & \, + \dim(\aff(\cup_{\kappa \in \E}\A_\kappa)) + 1\\
		\le &\, \sum_{\kappa \in \II}\dim(\A_\kappa) + |\II|- 1  +  \sum_{\kappa \in \II'}\dim(\A_\kappa) + |\II'|- 1  \\
		& \, +  \sum_{\kappa \in \E}\dim(\A_\kappa) + |\E|- 1 + 2\\
		=&\,\sum_{\ell=1}^{n}\dim(\A_{\ell}) + n - 1. \nonumber
	\end{split}
	\end{equation}
	From the definition of affine independence, we have that $\{\A_{\ell}\}_{\ell=1}^n$ is not affinely independent. This completes the proof by contrapositive.
\end{proof}

\subsection{Lemma for the Proof of Theorem \ref{thm:interior-subspace-preserving}}

We start by introducing an important lemma which will be used 
our proofs.
\begin{definition}[Strict convex combination]
A strict convex combination of a set of points $\{\y_i\}_{i=1}^N$ is a point of the form $\sum_{i=1}^N c_i \y_i$, where $c_i > 0$ and $\sum_{i=1}^N c_i= 1$.
\end{definition}

\begin{lemma}
	\label{lema:dense-convex-combination}
	Let $\{\y_i\}_{i=1}^N$ be an arbitrary set of data points. Any relative interior point of $\conv \{\y_1, \cdots, \y_N\}$ is a strict convex combination of $\{\y_i\}_{i=1}^N$.
\end{lemma}

\begin{proof}
We provide a constructive proof.

Let $\C = \conv \{\y_1, \cdots, \y_N\}$. Since $\x$ is a relative interior point of $\C$, it is known that
(see, e.g. \cite{Rockafellar:book98}) for each $\y_i$, there exists $\x_i \in \C$, $\x_i \ne \x$ such that $\x$ lies in the line segment between $\x_i$ and $\y_i$. That is, there exists $b_{i0} \in (0, 1]$ such that
\begin{equation}
	\x = b_{i0} \y_i + (1-b_{i0}) \x_i.
	\label{eq:prf-strict-convex-combination-interior}
\end{equation}
Moreover, since $\x_i \in \C$, we can further write $\x_i$ as a convex combination of $\{\y_i\}_{i=1}^N$. That is, there exists $b_{ij} \in [0, 1], \sum_{j=1}^N b_{ij} = 1$ such that $\x_i = \sum_{j=1}^N b_{ij} \y_j$. By substituting this expression of $\x_i$ into \eqref{eq:prf-strict-convex-combination-interior}, we can write $\x$ as
\begin{equation}
	\x = b_{i0} \y_i + (1 - b_{i0}) \sum_{j=1}^N b_{ij} \y_j.
	\label{eq:prf-strict-convex-combination-1}
\end{equation}
By taking the average of each side of \eqref{eq:prf-strict-convex-combination-1} over $i \in \{1, \cdots, N\}$, we get
\begin{equation}
\x = \sum_{i=1}^N\frac{b_{i0}}{N} \y_i + \sum_{i=1}^N \frac{(1 - b_{i0})}{N} \sum_{j=1}^N b_{ij} \y_j.
\label{eq:prf-strict-convex-combination-2}
\end{equation}
One can check that the right hand side of \eqref{eq:prf-strict-convex-combination-2} is a strict convex combination of points $\{\y_i\}_{i=1}^N$. This completes the proof.

\end{proof}

\subsection{Comparison of Corollary~\ref{thm:faces-subspace-preserving} and Theorem~\ref{thm:faces-subspace-preserving-sufficient}}
\label{sec:appendix-comparison-coro-IV2-thm-IV6}

We provide a proof for the argument that the sufficient condition in Theorem~\ref{thm:faces-subspace-preserving-sufficient} implies the sufficient condition in Corollary~\ref{thm:faces-subspace-preserving},
\ie,
if $\A_{\ell} \cap \conv (\X^{(-\ell)}) = \emptyset$ then $\aff(\F_{j}^{(\ell)}) \cap \conv (\X^{(-F_{j}^{(\ell)})}) = \emptyset$.

\begin{proof}
	We prove the contrapositive of the statement. Assume that $\aff(\F_{j}^{(\ell)}) \cap \conv (\X^{(-F_{j}^{(\ell)})}) \ne \emptyset$, and let $\y$ be a point that lies in such intersection.
	From $\y \in \conv (\X^{(-F_{j}^{(\ell)})})$ we have
	\begin{equation}
	\begin{split}
	\y = \sum_{k: \x_k \notin \F_{j}^{(\ell)}, \x_k \in \A_{\ell}} c_k \x_k + \sum_{k:\x_k \notin \A_{\ell}} c_k \x_k,
	\end{split}
	\label{eq:convex-neg-face}
	\end{equation}
where $\sum_{k: \x_k \notin \F_{j}^{(\ell)}, \x_k \in \A_{\ell}} c_k + \sum_{k:\x_k \notin \A_{\ell}} c_k = 1$,
	and $c_k \ge 0$ for all $k \in \{k: \x_k \notin \F_{j}^{(\ell)}, \x_k \in \A_{\ell}\} \cup \{k:\x_k \notin \A_{\ell}\}$.
	In addition, we can see that $\sum_{k:\x_k \notin \A_{\ell}} c_k \ne 0$ from the following argument.
	Suppose for the purpose of arriving at a contradiction that $\sum_{k:\x_k \notin \A_{\ell}} c_k = 0$, \ie, $c_k = 0$ for $k$ where $\x_k\notin \A_{\ell}$. Then we have $\y = \sum_{k: \x_k \notin \F_{j}^{(\ell)}, \x_k \in \A_{\ell}} c_k \x_k$, which implies that $\y \in \Q^{(\ell)}$ (since $\x_k \in \A_{\ell}$) and $\y \notin \F_{j}^{(\ell)}$  (since $\x_k \notin \F_{j}^{(\ell)}$).
This contradicts with the fact that $\y$ also lies in $\aff(\F_{j}^{(\ell)})$.
Therefore, we have shown that $\sum_{k:\x_k \notin \A_{\ell}} c_k \ne 0$.
	
	Now, from the fact that $\y$ also lies in $\aff(\F_{j}^{(\ell)})$ we have
	\begin{equation}
	\y = \sum_{k:\x_k \in \F_{j}^{(\ell)}} c_k' \x_k, \text{~~where~~} \sum_{k:\x_k \in \F_{j}^{(\ell)}} c_k' = 1.
	\label{eq:affine-face}
	\end{equation}
	Combining \eqref{eq:convex-neg-face} with \eqref{eq:affine-face} we have
	\begin{equation}
	\sum_{k: \x_k \notin \F_{j}^{(\ell)}, \x_k \in \A_{\ell}} c_k \x_k + \sum_{k:\x_k \notin \A_{\ell}} c_k \x_k = \sum_{k:\x_k \in \F_{j}^{(\ell)}} c_k' \x_k,
	\end{equation}
	which can be rearranged into
	\begin{equation}
	\sum_{k:\x_k \notin \A_{\ell}} c_k \x_k = \sum_{k:\x_k \in \F_{j}^{(\ell)}} c_k' \x_k - \sum_{k: \x_k \notin \F_{j}^{(\ell)}, \x_k \in \A_{\ell}} c_k \x_k.
	\end{equation}
	We divide both sides of the equation above by $\sum_{k:\x_k \notin \A_{\ell}} c_k$, which gives the following equation:
	\begin{multline}
	\sum_{k:\x_k \notin \A_{\ell}} \frac{c_k}{\sum_{i:\x_i \notin \A_{\ell}} c_i} \x_k = \sum_{k:\x_k \in \F_{j}^{(\ell)}} \frac{c_k'}{\sum_{i:\x_i \notin \A_{\ell}} c_i} \x_k \\- \sum_{k: \x_k \notin \F_{j}^{(\ell)}, \x_k \in \A_{\ell}} \frac{c_k}{\sum_{i:\x_i \notin \A_{\ell}} c_i} \x_k.
	\label{eq:affine-convex-intersection}
	\end{multline}
	One can check that the LHS of \eqref{eq:affine-convex-intersection} is a convex combination of points $\X^{(-\ell)}$.
 Since that $$\frac{\sum_{k:\x_k \in \F_{j}^{(\ell)}} c_k' - \sum_{k: \x_k \notin \F_{j}^{(\ell)}, \x_k \in \A_{\ell}} c_k}{\sum_{i:\x_i \notin \A_{\ell}}c_i}=1,$$ we have that the RHS of \eqref{eq:affine-convex-intersection} is an affine combination of points in $\A_{\ell}$.
	Therefore, we have shown that $\conv(\X^{(-\ell)}) \cap \A_{\ell} \ne \emptyset$. This completes the proof of contrapositive.
\end{proof}

\subsection{Proof for Theorem~\ref{thm:subspace-dense}}

\begin{proof}
When $\x_j \in \ri (\Q^{(\ell)} )$, by Lemma~\ref{lema:dense-convex-combination}, we can construct a strict convex combination of the columns in $X_{-j}^{(\ell)}$, with every coefficient being nonzero. 
\end{proof}

\subsection{Proof for Corollary \ref{corollary:collaborative-effect}}

\begin{proof}
Since $\c_j$ is a subspace-dense solution of $\x_j$ to \eqref{eq:affine-ssc}, we have that: $c_{ij} > 0$ for all $i=1,\cdots,j-1, j+1,\cdots, N_\ell$, where the index $i$ corresponds to other $N_\ell-1$ data points in $\A_\ell$. Then, the affinity $a_{ji} = a_{ij} = |c_{ij}| + |c_{ji}| \ge |c_{ij}| > 0$ for all $i=1,\cdots,j-1, j+1,\cdots, N_\ell$, which means that $\x_j$ connects to all the other $N_\ell-1$ data points in $\A_\ell$. 

\end{proof}

\subsection{Solving ASSC with ADMM}
\label{sec:grouping-effect-in-ADMM}


To accurately characterize the optimal solution in optimization practice, we translate all the data points at first, and then reformulate problem \eqref{eq:affine-ssc} as follows:
\begin{align}
\label{eq:affine-ssc-compact-translated}
\begin{split}
\min_{\c_j} \| \c_j \|_1, ~~ \st  ~~~ Y_{-j}\c_j =\0, ~~~ \1^\top \c_j =1.
\end{split}
\end{align}
where $Y_{-j} =X_{-j} - \x_j \1^\top$, and $X_{-j}$ is data matrix $X$ but excluding $\x_j$. To explicitly observe the role of the affine constraint in optimization, we introduce two Lagrangian multipliers for the two linear constraints, respectively. The Lagrangian dual problem is as follows
\begin{align}
\label{eq:affine-ssc-compact-translated-dual}
\begin{split}
\max_{\w_j,\nu_j} ~ - \nu_j  ~~\st  ~~ \| Y_{-j}^\top \w_j + \nu_j \1 \|_\infty \le 1,
\end{split}
\end{align}
where $\w_j \in \RR^D$ and $\nu_j \in \RR$ are the dual variables. Notice that problem \eqref{eq:affine-ssc} and \eqref{eq:affine-ssc-compact-translated} are equivalent, because the problem is translation-invariant due to the affine constraint.
%
\begin{lemma}[\bf dual certificate for nonnegative solution]
The optimal solution $\c_j^\ast$ of problem \eqref{eq:affine-ssc-compact-translated} is nonnegative if and only if the optimal dual variable $\nu_j^\ast$ is $-1$.
\label{lemma:dual-certificate-for-nonnegative-solution}
\end{lemma}
\begin{proof}
Note that the optimization problem in \eqref{eq:affine-ssc-compact-translated} is convex. By the strong duality theory, the dual gap between the primal problem \eqref{eq:affine-ssc-compact-translated} and its dual problem \eqref{eq:affine-ssc-compact-translated-dual}
is zero, \ie, $-\nu_j^\ast =\|\c_j^\ast\|_1$. When the optimal solution is nonnegative, we have that $\|\c_j^\ast\|_1 =\1^\top \c_j^\ast=1$ and thus $\nu_j^\ast=-1$.
\end{proof}

Note that problem \eqref{eq:affine-ssc-compact-translated} is convex but not strongly convex, and the $\ell_1$ norm regularizer becomes vacuous when the solution is nonnegative, thus the optimal solution is not unique. Because of the non-uniqueness of the nonnegative optimal solution, Theorem~{IV.6} states only the existence of the subspace-dense solution for data point in the relative interior. 
Without specifying a concrete algorithm, it is impossible to talk about 
the subspace-dense solution of ASSC.

Consider solving problem~\eqref{eq:affine-ssc-compact-translated} with ADMM. When the optimal solution $\c_j^\ast$ is nonnegative, 
we can show that the ``shrinkage'' effect of the shrinkage thresholding operator $\S_{\frac{1}{\mu}}$ is compensated by the term $-\frac{\nu^\ast}{\mu} \1$ where $\mu$ is a penalty parameter and $\nu^\ast =-1$.

Suppose the solution is subspace-preserving. Then, it is equivalent to consider the solution of the following problem
\begin{align}
\label{eq:fictitious-primal2}
\begin{split}
\min_{\c} \| \c \|_1 ~~\st  ~~~ Y_{-j}^{(\ell)} \c  = \0, ~~ \1^\top \c  =1,
\end{split}
\end{align}
where $Y_{-j}^{(\ell)}$ consists of the columns in subspace $\A_\ell$.
%
%
%
Note that problem~\eqref{eq:fictitious-primal2} is equivalent to: 
\begin{align}
\label{eq:fictitious-primal3}
\begin{split}
\min_{\c,\z} \| \z \|_1 ~~\st  ~~~ \z=\c, ~~ Y_{-j}^{(\ell)} \c  = \0, ~~ \1^\top \c  =1.
\end{split}
\end{align}

We solve the equivalent ASSC problem in~\eqref{eq:fictitious-primal3} via ADMM~\cite{Lin:09, Boyd:FTML10} and show explicitly how the affine constraint eliminates the sparsifying operation of the $\ell_1$ norm.

The augmented Lagrangian is as follows:
\begin{align}
\label{eq:fictitious-augmented-Lag}
\begin{split}
&\L(\c,\z,\y,\w,\nu) \\
&=\|\z\|_1 +\langle \y, \z-\c\rangle + \langle \w, Y_{-j}^{(\ell)} \c \rangle + \nu (\1^\top \c -1) \\
&+\frac{\mu}{2}(\|\z -\c\|^2 + \|Y_{-j}^{(\ell)}\c\|^2 + (\1^\top \c -1)^2).
\end{split}
\end{align}
where $\y$, $\w$, and $\nu$ are multipliers (i.e. the dual variables), and $\mu > 0$ is a penalty parameter.

We solve problem~\eqref{eq:fictitious-augmented-Lag} by alternatively updating each of the primal variables $\z$, $\c$, and the dual variables $\y$, $\w$, and $\nu$ by holding the other variables fixed.

\myparagraph{Updating $\z$}
When other variables are fixed, we update $\z$ as follows:
\begin{align}
\label{eq:updating-z}
\begin{split}
\z_{t+1}
        =\S_{\frac{1}{\mu}}(\c_t - \frac{\y_t}{\mu_t}),
\end{split}
\end{align}
where $\S_\tau(\cdot)$ is a shrinkage shresholding operator~\cite{Beck2009} with parameter $\tau$.

\myparagraph{Updating $\c$}
When other variables are fixed, we update $\c$ as follows:
\begin{align}
\label{eq:updating-c}
\begin{split}
\c_{t+1} 
        =(I\! +\! Y_{-j}^{(\ell)\top} Y_{-j}^{(\ell)}+ \!\1 \1^\top)^{-1}(\z_{t} \!+ \! \1 \!+\! \frac{\y_t \!-\!\nu_t \1 \!-\! Y_{-j}^{(\ell)\top} \w_t}{\mu_t})
\end{split}
\end{align}
where $I$ is an identity matrix.

\myparagraph{Updating the dual variables $\y$, $\w$, and $\nu$}
\begin{align}
\label{eq:updating-y-w-v}
\begin{split}
\y_{t+1}=& \y_t + \mu_t(\z_{t+1}-\c_{t+1}) \\
\w_{t+1}=& \w_t + \mu_t(Y_{-j}^{(\ell)}\c_{t+1}) \\
\nu_{t+1}=& \nu_t + \mu_t(\1^\top \c_{t+1} -1)
\end{split}
\end{align}

\myparagraph{Updating the penalty parameter $\mu$}
\begin{align}
\label{eq:updating-mu}
\begin{split}
\mu_{t+1}= \mu_t \rho,
\end{split}
\end{align}
where $\rho >1$.

Notice that problem~\eqref{eq:fictitious-primal3} is convex and feasible. The ADMM algorithm will converge to its optimal solution. Suppose that when converged we have that:
\begin{align}
\label{eq:ADMM-converge-z-c-y-w-v}
\begin{split}
\z_{t+1}&=~\c^\ast, ~~\c_{t+1}=\c^\ast, ~~\1^\top\c^\ast =~1, ~~Y_{-j}^{(\ell)}\c^\ast=\0, \\
\y_{t+1}&=\y_\ast, ~~\w_{t+1}=\w_\ast, ~~\nu_{t+1}=\nu_\ast, ~ \mu_t=\mu_\ast,
\end{split}
\end{align}

By using the facts that $\c_{t+1}=\c^\ast$, $\z_{t+1}= \c^\ast$, and $Y_{-j}^{(\ell)}\c^\ast = \0$, when the algorithm converged, the updating rule~\eqref{eq:updating-c} for $\c_{t+1}$ turns out to be:
\begin{align}
\label{eq:updating-c-when-converged}
\begin{split}
\c^\ast=(I +& Y_{-j}^{(\ell)\top} Y_{-j}^{(\ell)} + \1 \1^\top)^{-1}\\
            &(\c^\ast + \1 + \frac{\y_\ast -\nu_\ast \1 - Y_{-j}^{(\ell)\top} \w_\ast}{\mu_\ast}).
\end{split}
\end{align}
Then we can see that
\begin{align}
\label{eq:updating-y-when-converged}
\begin{split}
\y_\ast =\mu_\ast [(Y_{-j}^{(\ell)\top} Y_{-j}^{(\ell)} + \1 \1^\top)\c^\ast -\1]+\nu_\ast \1+ Y_{-j}^{(\ell)\top} \w_\ast.
\end{split}
\end{align}
Therefore we have that
\begin{align}
\label{eq:updating-z-when-converged-c-y}
\begin{split}
\c^\ast -\frac{1}{\mu_\ast}\y_\ast = \c^\ast - [(Y_{-j}^{(\ell)\top} Y_{-j}^{(\ell)} + \1 \1^\top)\c^\ast -\1] \\
- \frac{\nu_\ast}{\mu_\ast}\1 - \frac{Y_{-j}^{(\ell)\top} \w_\ast}{\mu_\ast}.
\end{split}
\end{align}
When the optimal solution $\c^\ast$ is nonnegative, we must have that $\|\c^\ast\|_1=\1^\top \c^\ast =1$. From the strong duality, we can conclude that $\nu_\ast =-1$. After some calculations, we can observe that the shrinkage thresholding operator could be compensated by the last term $- \frac{\nu_\ast}{\mu_\ast}\1$ in Eq.~\eqref{eq:updating-z-when-converged-c-y}.

By substituting Eq.~\eqref{eq:updating-z-when-converged-c-y} into Eq.~\eqref{eq:updating-z}, we have that
\begin{align}
\label{eq:updating-c-when-converged-without-sparsify}
\begin{split}
\z_{t+1} &=\S_{\frac{1}{\mu_\ast}}(\c^\ast -\frac{1}{\mu_\ast}\y_\ast) \\
         &= \c^\ast - (Y_{-j}^{(\ell)\top} Y_{-j}^{(\ell)} + \1 \1^\top)\c^\ast +\1-\frac{Y_{-j}^{(\ell)\top} \w_\ast}{\mu_\ast}\\
         &= \c^\ast - Y_{-j}^{(\ell)\top} Y_{-j}^{(\ell)} \c^\ast -\frac{Y_{-j}^{(\ell)\top} \w_\ast}{\mu_\ast}.
\end{split}
\end{align}
Surprisingly, the shrinkage thresholding operator disappears.



\section*{Acknowledgment}
C.-G. Li was partially supported by the National Natural Science Foundation of China under Grant No. 61876022, and the Open Project Fund from Key Laboratory of Machine Perception (MOE), Peking University. C. You and R. Vidal were supported by the National Science Foundation under Grant No. 1447822.



%

\bibliographystyle{IEEEbib}
\bibliography{biblio/collaborator/cgli,biblio/temp,biblio/vidal,biblio/vision,biblio/math,biblio/learning,biblio/consensus,biblio/sparse,biblio/geometry,biblio/dti,biblio/recognition,biblio/surgery,biblio/coding,biblio/matrixcompletion,biblio/segmentation}
%






\end{document}